\documentclass[a4paper,10pt]{article}
\usepackage[utf8]{inputenc}
\usepackage{amsthm}
\usepackage{amsfonts}
\usepackage{amssymb}	
\usepackage{amsmath}
\allowdisplaybreaks
\usepackage{mathtools}
\usepackage[english]{babel}
\usepackage{color}
\usepackage{cite}
\usepackage[numbers]{natbib}
\usepackage{slashed}
\usepackage{enumerate}
\usepackage{graphicx}
\usepackage{systeme}
\usepackage{dsfont}
\usepackage{relsize}
\usepackage[multiple]{footmisc}
\usepackage[margin=1in]{geometry}
\usepackage{float}
\usepackage{subcaption}
\usepackage{hyperref}
\usepackage[labelformat=simple]{subcaption}

\usepackage{graphicx}
\usepackage{subcaption}
\usepackage{bmpsize}
\usepackage{epstopdf}

\usepackage{xcolor}
\usepackage{hyperref}

\newcommand{\R}{\mathbb{R}}
\newcommand{\C}{\mathbb{C}} 

\renewcommand{\d}{\partial}
\renewcommand{\div}{{\rm div}}

\newcommand{\kabs}{|k|}

\renewcommand{\Re}{{\rm Re}}
\newcommand{\idmat}{\mathds{1}}

\newcommand{\op}{(1)}
\newcommand{\twop}{(2)}
\newcommand{\thp}{(3)}
\newcommand{\fp}{(4)}
\newcommand{\fip}{(5)}
\newcommand{\sixp}{(6)}
\newcommand{\sevp}{(7)}
\newcommand{\eip}{(8)}

\newcommand{\mue}{\mu_e}
\newcommand{\muc}{\mu_c}
\newcommand{\mumic}{\mu_{\text{micro}}}
\newcommand{\lame}{\lambda_e}
\newcommand{\lammic}{\lambda_{\text{micro}}}
\newcommand{\Lc}{L_c^2}

\newcommand{\norm}[1]{\left\lVert#1\right\rVert}

\newcommand{\sigmatil}{\widetilde{\sigma}}
\newcommand{\sigmahat}{\widehat{\sigma}}

\usepackage{empheq}

\usepackage{mdframed}
\usepackage{lipsum} 
\usepackage[many]{tcolorbox}

\newcommand{\Ce}{\mathbb{C}_e}
\newcommand{\Te}{\mathbb{T}_e}
\newcommand{\Tc}{\mathbb{T}_c}
\newcommand{\Cc}{\mathbb{C}_c}

\newcommand{\Jmic}{\mathbb{J}_{\text{micro}}}
\newcommand{\Cmic}{\mathbb{C}_{\text{micro}}}
\newcommand{\Jc}{\mathbb{J}_c}

\newcommand{\Le}{\mathbb{L}_e}

\newcommand{\Cet}{\widetilde{\mathbb{C}}_e}

\theoremstyle{remark}

\newtheorem{lemma}{Lemma}

\DeclareMathOperator{\tr}{tr}

\DeclareMathOperator{\curl}{curl}
\DeclareMathOperator{\Div}{Div}
\DeclareMathOperator{\Curl}{Curl}
\DeclareMathOperator{\sym}{sym}
\let\skew\relax

\DeclareMathOperator{\skew}{skew}

\numberwithin{equation}{section}

\makeatletter

\let\@fnsymbol\@arabic

\title{Relaxed micromorphic broadband scattering for finite-size meta-structures - a detailed development}
\author{Alexios Aivaliotis\thanks{Alexios Aivaliotis, alexios.aivaliotis@insa-lyon.fr, GEOMAS, INSA-Lyon, Universit\'e de Lyon, 20 avenue Albert Einstein,
		69621, Villeurbanne cedex, France}~, Domenico Tallarico\thanks{Domenico Tallarico, domenico.tallarico@empa.ch, Empa, Swiss Federal Laboratories for Materials Science and Technology, Lab. Acoustics/Noise Control, \"Uberlandstrasse 129, 8600 D\"ubendorf, Switzerland}~, Marco-Valerio d'Agostino\thanks{Marco-Valerio d'Agostino, marco-valerio.dagostino@insa-lyon.fr, GEOMAS, INSA-Lyon, Universit\'e de Lyon, 20 avenue Albert Einstein,
		69621, Villeurbanne cedex, France}~, Ali Daouadji\thanks{Ali Daouadji, ali.daouadji@insa-lyon.fr, GEOMAS, INSA-Lyon, Universit\'e de Lyon, 20 avenue Albert Einstein,
		69621, Villeurbanne Cedex, France}~,\\ Patrizio Neff\thanks{Patrizio Neff, patrizio.neff@uni-due.de, Head of Chair for Nonlinear Analysis and Modelling, Fakultät für Mathematik, Universität Duisburg-Essen,
		Mathematik-Carrée, Thea-Leymann-Straße 9, 45127 Essen, Germany}~ and Angela Madeo\thanks{Angela Madeo, corresponding author, angela.madeo@insa-lyon.fr, GEOMAS, INSA-Lyon, Universit\'e
		de Lyon, 20 avenue Albert Einstein, 69621, Villeurbanne Cedex, France}}
\date{\today}

\begin{document}
\maketitle
\small
\begin{abstract}
The conception of new metamaterials showing unorthodox behaviors with respect to elastic wave propagation has become possible in recent years thanks to powerful modeling tools, allowing the description of an infinite medium generated by periodically architectured base materials. Nevertheless, when it comes to the study of the scattering properties of finite-sized structures, dealing with the correct boundary conditions at the macroscopic scale becomes challenging. In this paper, we show how finite-domain boundary value problems can be set-up in the framework of enriched continuum mechanics (relaxed micromorphic model) by imposing continuity of macroscopic displacement and of generalized tractions, as well as additional conditions on the microdistortion tensor and on the double-traction.
	
The case of a metamaterial slab of finite width is presented, its scattering properties are studied via a semi-analytical solution of the relaxed micromorphic model and compared to finite element simulations encoding all details of the selected microstructure. The reflection coefficient obtained via the two methods is presented as a function of the frequency and of the direction of propagation of the incident wave. We find excellent agreement for a large range of frequencies going from the long-wave limit to frequencies beyond the first band-gap and for angles of incidence ranging from normal to near parallel incidence. The case of a semi-infinite metamaterial is also presented and is seen to be a reliable measure of the average behavior of the finite meta-structure. A tremendous gain in terms of computational time is obtained when using the relaxed micromorphic model for the study of the considered meta-structure. The present paper is an important stepping stone for the development of a theoretical paradigm shift enabling the exploration of meta-structures at large scales.
\end{abstract}
\normalsize
\textbf{Keywords:} enriched continuum mechanics, anisotropic metamaterials, band-gaps, wave-propagation, relaxed micromorphic model, interface, scattering, finite-sized meta-structures.

\vspace{2mm}
\textbf{}\\
\textbf{AMS 2010 subject classification:} 74A30 (nonsimple
materials), 74A60 (micromechanical theories),
74B05 (classical linear elasticity), 74J05 (linear waves), 74J10 (bulk waves), 75J15 (surface waves), 74J20 (wave scattering), 74M25 (micromechanics), 74Q15
(effective constitutive equations).

\newpage
\tableofcontents

\section{Introduction}\label{sec:intro}

Recent years have seen the rapid development of mechanical metamaterials and phononic crystals showing unorthodox behaviors with respect to elastic wave propagation, including focusing, channeling, cloaking, filtering, etc. \cite{misseroni2016cymatics,chen2010acoustic,craster2012acoustic,kadic2013metamaterials,norris2015acoustic}. The basic idea underlying the design of these metamaterials is that of suitably engineering the architecture of their microstructure in such a way that the resulting macroscopic (homogenized) properties can exhibit the desired exotic characteristics. One of the most impressive features provided by such metamaterials is that of showing band-gaps, \emph{i.e.} frequency ranges for which wave propagation is inhibited. The most widespread class of metamaterials consists of those which are obtained by a periodic repetition in space of a specific unit cell and which are known as periodic metamaterials. For such metamaterials, renown scientists have provided analytical \cite{willis1997dynamics,willis2009exact,willis2011effective,willis2012construction} or numerical \cite{geers2010multi} homogenization techniques (based on the seminal works of Bloch \cite{bloch1929quantenmechanik} and Floquet \cite{floquet1883equations}), allowing to obtain a homogenized model which suitably describes, to a good extent, the dynamical behavior of the bulk periodic metamaterial at the macroscopic scale. Rigorous models for studying the macroscopic mechanical behavior of non-periodic structures become rarer and usually rely on the use of detailed finite element modeling of the considered microstructures (see e.g. \cite{krushynska2014towards}), thus rendering the implementation of large structures computationally demanding and impractical. 

At the current state of knowledge, little effort is made to model large metamaterial structures (which we will call \emph{meta-structures}), due to the difficulty of imposing suitable boundary conditions in the framework of homogenization theories (see \cite{srivastava2017evanescent}). In this article, we shed new light in this direction by the  introduction of an enriched continuum model of the micromorphic type (relaxed micromorphic model), equipped with the proper boundary conditions for the effective description of a $2$D finite-sized band-gap meta-structure. The relaxed micromorphic model (see \cite{aivaliotis2018low,barbagallo2019relaxed,neff2014unifying,dagostino2019effective,madeo2014band,madeo2015wave,madeo2016complete,madeo2016first,madeo2016reflection,madeo2017modeling,madeo2017relaxed,madeo2017review,madeo2017role,barbagallo2017transparent,neff2015relaxed,neff2017real} for preliminary results)
has a suitable structure which allows to describe the average properties of (periodic or even non-periodic) anisotropic metamaterials with a limited number of constant material parameters and for an extended frequency range going from the long-wave limit to frequencies which are beyond the first band-gap. The rigorous development of the relaxed micromorphic model for anisotropic metamaterials has been given in \cite{dagostino2019effective,neff2019identification},  where applications to different classes of symmetry and the particular case of tetragonal periodic metamaterials are also discussed. In \cite{dagostino2019effective}, a procedure to univocally determine some of the material parameters for periodic metamaterials with static tests is also provided. It is important to point out that the relaxed micromorphic model is not obtained via a formal homogenization procedure, but is developed generalizing the framework of macroscopic continuum elasticity by introducing enriched kinematics and enhanced constitutive laws for the strain energy density. In this way, extra degrees of freedom are added to the classical macroscopic displacement via the introduction of the micro-distortion tensor and the chosen constitutive form for the anisotropic strain energy density. This allows to introduce a limited number of elastic parameters through fourth-order macro and micro elasticity tensors working on the sym/skew orthogonal decomposition of the introduced deformation measures (see \cite{dagostino2019effective} for details). The need of using an enriched continuum model of the micromorphic type for describing the broadband macroscopic behavior of acoustic metamaterials as emerging from a numerical homogenization technique was recently shown in \cite{sridhar2018general}. Nevertheless, the authors of \cite{sridhar2018general} show that a huge number of elastic parameters (up to 600 for the studied tetragonal two-dimensional structures) is indeed needed to perform an accurate fitting of the dispersion curves issued by the Bloch-Floquet analysis. This extensive number of parameters can also be found in other micromorphic models of the Eringen \cite{eringen1999microcontinuum} and Mindlin \cite{mindlin1964micro} type. The need of this vast number of parameters is related to the fact that the macroscopic class of symmetry of the metamaterial and the correct (\emph{i.e.} sym/skew-decomposed) deformation measures are usually not accounted for. The fitting of the dispersion curves, which can be obtained by the procedure proposed in \cite{dagostino2019effective}, cannot reproduce point-by-point the dispersion curves issued via Bloch-Floquet analysis (which is not the aim of our work), but is general enough to capture the main features of the studied metamaterials' behavior including dispersion, anisotropy, band-gaps for a wide range of frequencies and for wavelengths which can become very small and even comparable to the size of the unit cell. A more refined fitting of the dispersion curves, as well as the possibility of including in the modeling additional symmetries via odd and higher order elasticity tensors, calls for the establishment of a new comprehensive model of the micromorphic type. Given the results presented in this paper for the particular case of tetragonal symmetry, we can claim that the relaxed micromorphic model is an important starting point for the development of such a model.  

Stemming from a variational procedure, the relaxed micromorphic model is naturally equipped with the correct macroscopic boundary conditions which have to be applied on the boundaries of the considered metamaterials. This implies that the global refractive properties of metamaterials' boundaries can be described in the simplified framework of enriched continuum mechanics, thus providing important information while keeping simple enough to allow important computational time-saving. Very complex phenomena take place when an incident elastic wave hits a metamaterial's boundary, resulting in reflected and transmitted waves which can be propagative or evanescent depending on the frequency and angle of incidence of the incident wave itself. The primordial importance of evanescent (non-propagative) waves for the correct formulation of boundary value problems for metamaterials has been highlighted in \cite{srivastava2017evanescent,willis2016negative}, where the need of infinite evanescent modes for obtaining continuity of displacement and of tractions at the considered metamaterial's boundary is pointed out. One of the advantages of the enriched continuum model proposed in the present paper is that the generalized boundary conditions (continuity of macroscopic displacement and of generalized traction) associated to the considered boundary value problem are exactly fulfilled with a finite number of modes.  

We will show that the proposed framework allows to describe, to a good extent, the overall behavior of the reflection and transmission coefficients (generated by a plane incident wave) of a metamaterial slab of finite width treated as an inclusion between two semi-infinite homogeneous media. We obtain the scattering properties of the slab as a function of the frequency and angle of incidence of the plane incident wave. To the authors' knowledge, a boundary value problem which describes the broadband and dynamical behavior of realistic finite-size metamaterial structures via the introduction of enriched macroscopic broad-angle boundary conditions, is presented here for the first time.

The results show very good agreement (for a wide range of frequencies extending from the low-frequency Cauchy limit to frequencies beyond the first band-gap and for all the possible angles of incidence) with the direct FEM numerical implementation of the considered system in COMSOL, where the detailed geometry of the unit cell has been implemented in the framework of classical linear elasticity. We observe a tremendous advantage in terms of the computational time needed to perform the numerical simulations (few hours for the relaxed micromorphic model against weeks for the direct FEM simulation).

The paper is structured as follows: In Section \ref{sec:intro} we present an introduction together with the notation used throughout the paper. Section \ref{sec:governing_eqns} briefly recalls the governing equations describing the motion of Cauchy and relaxed micromorphic continua, with specific reference to the definition of the energy flux for both cases. In Section \ref{sec:wave_prop} the plane-wave solutions for the Cauchy and relaxed micromorphic continua are obtained as solutions of the corresponding eigenvalue problems. In Section \ref{sec:BCs} we provide essential information concerning the correct boundary conditions which have to be imposed at the metamaterial's boundaries, in the relaxed micromorphic framework. In Sections \ref{sec:RTSingle} and \ref{sec:RTSlab} the problems of the scattering from a relaxed micromorphic single interface and relaxed micromorphic slab of finite size, respectively, are rigorously set up and solved. Section \ref{sec:Comsol} presents the detailed implementation of the microstructured metamaterial slab based on classical linear elasticity and implemented in the commercial Finite Element software COMSOL Multiphysics. Section \ref{sec:results} thoroughly presents our results by means of a detailed discussion. Section \ref{sec:conclusions} is devoted to conclusions and perspectives.

\subsection{Notation}
Let $\R^{3\times 3}$ be the set of all real $3\times 3$ second order tensors which we denote by capital letters. A simple and a double contraction between tensors of any suitable order is denoted by $\cdot$ and $:$ respectively, while the scalar product of tensors of suitable order is denoted by $\left\langle \cdot, \cdot \right\rangle$.\footnote{For example, $(A\cdot v)_i=A_{ij}v_j$,  $(A\cdot B)_{ik}=A_{ij}B_{jk}$, $A:B=A_{ij}B_{ji}$, $(C\cdot B)_{ijk}=C_{ijp}B_{pk}$, $(C:B)_{i}=C_{ijp}B_{pj}$, $\left\langle v,w\right\rangle = v\cdot w = v_i w_i$, $\left\langle A, B \right\rangle = A_{ij}B_{ij}$, etc.} The Einstein sum convention is implied throughout this text unless otherwise specified. The standard Euclidean scalar product on $\R^{3 \times 3}$ is given by $\left\langle X, Y \right\rangle= \tr (X\cdot Y^T)$ and consequently the Frobenius tensor norm is $\norm{X}^2 = \left\langle X, X\right\rangle$. The identity tensor on $\R^{3\times 3}$ will be denoted by $\idmat$; then, $\tr(X)=\left\langle X, \idmat \right\rangle$.

We denote by $B_L$ a bounded domain in $\R^3$, by $\d B_L$ its regular boundary and by $\Sigma$ any material surface embedded in $B_L$. The outward unit normal to $\d B_L$ will be denoted by $\nu$ as will the outward unit normal to a surface $\Sigma$ embedded in $B_L$. Given a field $a$ defined on the surface $\Sigma$, we define the jump of $a$ through the surface $\Sigma$ as:
\begin{equation}
[[a]] = a^+ - a^-, \qquad \text{with} \qquad a^{-} := \lim_{\substack{x \in B_L^{-}\setminus \Sigma \\ x \to \Sigma}} a ,\qquad \text{and} \qquad a^{+} := \lim_{\substack{x \in B_L^{+}\setminus \Sigma \\ x \to \Sigma}} a,
\end{equation} 
where $B_L^{-}, B_L^{+}$ are the two subdomains which result from splitting $B_L$ by the surface $\Sigma$. 

Classical gradient $\nabla$ and divergence $\Div$ operators are used throughout the paper.\footnote{The operators $\nabla$, $\curl$ and $\Div$ are the classical gradient, curl and divergence operators. In symbols, for a field $u$ of any order, $(\nabla u)_i=u_{,i}$, for a vector field $v$, $(\curl v)_i = \epsilon_{ijk}v_{k,j}$ and for a field $w$ of order $k>1$, $(\Div w)_{i_1 i_2\ldots i_{k-1}} = w_{i_1 i_2\ldots i_k,i_k}$
	
} Moreover, we introduce the $\Curl$ operator of the matrix $P$ as $\left(\Curl P\right)_{ij}=\epsilon_{jmn} P_{in,m}$, where $\epsilon_{jmn}$ denotes the standard Levi-Civita tensor, which is equal to $+1$, if $(j,m,n)$ is an even permutation of $(1,2,3)$, to $-1$, if $(j,m,n)$ is an odd permutation of $(1,2,3)$, or to $0$ if $j=m$, or $m=n$, or $n=j$. The subscript $,j$ indicates derivation with respect to the $j-$th component of the space variable, while the subscript $,t$ denotes derivation with respect to time.

Given a time interval $[0,T]$, the classical macroscopic displacement field is denoted by:
\begin{equation} \label{eq:displacement} 
u(x,t)=(u_1(x,t),u_2(x,t),u_3(x,t)	)^{\rm T}, \qquad x\in B_L, \quad t \in [0,T].
\end{equation}
In the framework of enriched continuum models of the micromorphic type, extra degrees of freedom are added through the introduction of the micro-distortion tensor $P$ denoted by: 
\begin{equation}\label{eq:timeharmonic2}
P(x,t) = \left(\begin{array}{ccc}
P_{11}(x,t) & P_{12}(x,t) & P_{13}(x,t)\\
P_{21}(x,t) & P_{22}(x,t) & P_{23}(x,t)\\
P_{31}(x,t) & P_{32}(x,t) & P_{33}(x,t)
\end{array}\right), \qquad x\in B_L,\,\quad t \in [0,T].
\end{equation}

\section{Governing equations and energy flux}\label{sec:governing_eqns}
\subsection{The classical isotropic Cauchy continuum}
The equations of motions in strong form for a classical Cauchy continuum are:
\begin{equation}\label{eq:Cauchystrong}
\rho \, u_{,tt} = \Div \sigma, \qquad \rho \, u_{i,tt} = \sigma_{ij,j},
\end{equation}
where 
\begin{equation}\label{eq:Cauchystress}
\sigma = 2 \mu \sym \nabla u + \lambda \tr(\sym \nabla u)\idmat, \qquad \sigma_{ij} = \mu(u_{i,j}+u_{j,i})+\lambda u_{k,k}\delta_{ij},
\end{equation}
is the classical symmetric Cauchy stress tensor for isotropic materials and $\mu$ and $\lambda$ are the classical Lam\'e parameters.

The mechanical system we are considering is conservative and, therefore, the energy must be conserved in the sense that the following differential form of a continuity equation must hold:
\begin{equation}
E_{,t}+\div H=0,\label{eq:EnergyConservation}
\end{equation}
where $E$ is the total energy of the system and $H$ is the energy flux vector, whose explicit expression is given by (see e.g. \cite{aivaliotis2018low} for a detailed derivation) 
\begin{equation}\label{eq:Cauchyflux}
H = -\sigma \cdot u_{,t}, \qquad H_i = - \sigma_{ik}\, u_{k,t}.
\end{equation}

\subsection{The anisotropic relaxed micromorphic model}
The kinetic energy density of the anisotropic relaxed micromorphic model considered here is \cite{dagostino2019effective}:
\begin{align}
J(u_{,t},\nabla u_{,t},P_{,t})=\,&\frac{1}{2}\rho\,\langle u_{,t},u_{,t} \rangle + \frac{1}{2}\langle \Jmic \sym P_{,t}, \sym P_{,t} \rangle + \frac{1}{2} \langle\Jc \skew P_{,t}, \skew P_{,t} \rangle \nonumber \\
&+ \frac{1}{2}\langle \Te \sym \nabla u_{,t}, \sym \nabla u_{,t} \rangle+ \frac{1}{2}\langle \Tc \skew \nabla u_{,t} , \skew \nabla u_{,t} \rangle, \label{eq:kineticRMM}
\end{align}
where $u$ is the macroscopic displacement field, $P\in \R^{3\times 3}$ is the non-symmetric micro-distortion tensor and $\rho$ is the apparent macroscopic density. Moreover, $\Jmic, \Jc, \Te$ and $\Tc$ are $4$th order inertia tensors, whose precise form will be specified in the following. An existence and uniqueness result for models with these types of generalized inertia terms is given in \cite{owczarek2019nonstandard}.

The strain energy density for an anisotropic relaxed micromorphic medium is given by \cite{dagostino2019effective, neff2019identification}:\footnote{We could consider more complex expressions for the curvature term, which would also include anisotropy and in fact, such expressions will be explored in further works. Here we want to show that curvature terms provide small corrections to the overall behavior of the metamaterial and so, we limit ourselves to a simplified isotropic expression.}
\begin{align}
W(\nabla u, P, \Curl P)= \,&\frac{1}{2}\langle \Ce \sym (\nabla u - P), \sym (\nabla u - P )\rangle + \frac{1}{2}\langle \Cmic \sym P,\sym P\rangle \nonumber\\
&+\frac{1}{2}\langle \Cc \skew (\nabla u - P), \skew (\nabla u - P)\rangle 
+\frac{\mu_{\rm macro}\,L_c^2}{2}\,\langle \Curl P,  \Curl P\rangle, \label{eq:strainRMM}
\end{align}
where $\Ce, \Cmic,\Cc, \Le$ and $\mathbb{L}_c$ are $4$th order elasticity tensors, whose precise form will be given in the following and $\mu_{\rm macro}$ is a scalar stiffness weighting the curvature term. Moreover, $L_c$ is a characteristic length, which may account for size-effects in the metamaterial. These size-effects are known to be important for metamaterials in the static regime (see e.g. \cite{rokovs2019micromorphic}), but are not usually explored in dynamics. 

In this paper, we start exploring the role which $L_c$ may play in dynamical size-effects. Nevertheless, these results are preliminary and new modeling tools will need to be developed in order to comprehensively unveil non-local effects arising in finite-sized metamaterials.

Let $t_0>0$ be a fixed time and consider a bounded domain $B_L \subset \R^3$. The action functional of the system at hand is defined as 
\begin{equation}\label{eq:actionfunct}
\mathcal{A}=\int_0^{t_0} \int_{B_L} (J-W) dX dt,
\end{equation}
where $J$ is the kinetic and $W$ the strain energy of the system, defined by \eqref{eq:kineticRMM} and \eqref{eq:strainRMM}, respectively. The  minimization of the action functional  provides the governing equations for the anisotropic relaxed micromorphic model \cite{dagostino2019effective}:

\begin{align}
\rho \, u_{,tt} -\Div\sigmahat_{,tt}\, &= \Div \widetilde{\sigma}, \nonumber\\
\Jmic \sym P_{,tt} = \widetilde{\sigma}_e -s -\sym \Curl m, &\qquad \Jc \skew P_{,tt} =\widetilde{\sigma}_c - \skew \Curl m,  \label{eq:governingRMMclosed}
\end{align}
where we set
\small
\begin{align}
 \quad \widetilde{\sigma}_e &=\Ce \sym (\nabla u - P), \qquad \widetilde{\sigma}_c= \Cc \skew (\nabla u - P), \qquad \widetilde{\sigma}=\widetilde{\sigma}_e+\widetilde{\sigma}_c, \qquad \sigmahat = \Te \sym \nabla u + \Tc \skew \nabla u,  \nonumber\\
s&=\Cmic \sym P, \qquad m:=\mu_{\rm macro}~L_c^2 \,\Curl P. \label{eq:relaxedRHS2}
\end{align}
\normalsize

Conservation of energy for an anisotropic relaxed micromorphic continuum is formally written as in equation \eqref{eq:EnergyConservation}. The specific form for the energy flux for an anisotropic relaxed micromorphic continuum is (see Appendix \ref{app:derivationflux} for detailed derivation of this expression):
\begin{equation}\label{eq:fluxAniso}
H = -\left(\sigmatil + \sigmahat\right)^T\cdot u_{,t} -\left( m^T \cdot P_{,t} \right):\epsilon,\quad H_k = -u_{i,t}\left(\sigmatil_{ik}+\sigmahat_{ik}\right)- m_{ih}P_{ij,t}\epsilon_{jhk}.
\end{equation}

\subsection{The plane-strain tetragonal symmetry case}\label{sec:tetragonal}
We are interested in plane-strain solutions of the system \eqref{eq:governingRMMclosed}. This means that the macroscopic displacement $u$ and the micro-distortion tensor $P$ introduced in \eqref{eq:displacement} and \eqref{eq:timeharmonic2} are supposed to take the following form:
\begin{equation}\label{eq:plane-strain-u}
u = u(x_1,x_2) = \left(u_1(x_1,x_2),u_2(x_1,x_2),0\right)^{\rm T}, \quad \text{and} \quad u_{i,3}=0, \quad i=1,2,
\end{equation}
and
\begin{equation}\label{eq:plane-strain-P}
P= P(x_1,x_2) =  \begin{pmatrix}
P_{11}(x_1,x_2)&P_{12}(x_1,x_2)&0\\
P_{21}(x_1,x_2)&P_{22}(x_1,x_2)&0\\
0&0&0\\
\end{pmatrix},~~{\rm and}~~P_{ij,3}=0, \quad i=1,2, ~j=1,2.
\end{equation}

Following \cite{dagostino2019effective}, we denote the second order constitutive tensors in Voigt notation corresponding to the fourth order ones appearing in the kinetic and strain energy expressions \eqref{eq:kineticRMM} and \eqref{eq:strainRMM} by a tilde.\footnote{For example, the $4$th order tensor $\Ce$ is written as $\Cet$ in Voigt notation.} For the tetragonal case, the constitutive tensors in Voigt notation are (see \cite{dagostino2019effective,barbagallo2017transparent}):
\small

\begin{align}
\Cet &= \left( 
\begin{array}{cccccc}
2\,\mue + \lame & \lame & \star& 0 & 0 & 0 \\
\lame & 2\,\mue + \lame & \star & 0 & 0 & 0 \\
\star & \star & \star & 0 & 0 & 0 \\
0 & 0 & 0 & \star & 0 & 0 \\
0 & 0 & 0 & 0 &\star & 0 \\
0 & 0 & 0 & 0 & 0 & \mue^{\ast} 
\end{array}\right), \quad 
\widetilde{\mathbb{C}}_{\rm micro} = 
\left( 
\begin{array}{cccccc}
2\,\mumic + \lammic & \lammic & \star & 0 & 0 & 0 \\
\lammic & 2\,\mumic + \lammic & \star & 0 & 0 & 0 \\
\star & \star & \star & 0 & 0 & 0 \\
0 & 0 & 0 & \star & 0 & 0 \\
0 & 0 & 0 & 0 & \star & 0 \\
0 & 0 & 0 & 0 & 0 & \mumic^{\ast} 
\end{array}\right), \nonumber \\
\widetilde{\mathbb{C}}_c &= \left( 
\begin{array}{ccc}
\star & 0 & 0 \\
0 & \star & 0 \\
0 & 0 & 4\,\muc 
\end{array}
\right), \qquad \qquad
\widetilde{\mathbb{J}}_{\rm micro} = 
\left( 
\begin{array}{cccccc}
2\,\eta_1 + \eta_3 & \eta_3 & \star & 0 & 0 & 0 \\
\eta_3 & 2\,\eta_1 + \eta_3 & \star & 0 & 0 & 0 \\
\star & \star & \star & 0 & 0 & 0 \\
0 & 0 & 0 & \star & 0 & 0 \\
0 & 0 & 0 & 0 & \star & 0 \\
0 & 0 & 0 & 0 & 0 & \eta_1^{\ast} 
\end{array}\right), \qquad  \qquad
\widetilde{\mathbb{J}}_c = \left( 
\begin{array}{ccc}
\star & 0 & 0 \\
0 & \star & 0 \\
0 & 0 & 4\,\eta_2 
\end{array}
\right), \nonumber \\
\widetilde{\mathbb{T}}_{e} &= 
\left( 
\begin{array}{cccccc}
2\,\bar{\eta}_1 + \bar{\eta}_3 & \bar{\eta}_3 & \star & 0 & 0 & 0 \\
\bar{\eta}_3 & 2\,\bar{\eta}_1 + \bar{\eta}_3 & \star & 0 & 0 & 0 \\
\star & \star & \star & 0 & 0 & 0 \\
0 & 0 & 0 & \star & 0 & 0 \\
0 & 0 & 0 & 0 & \star & 0 \\
0 & 0 & 0 & 0 & 0 & \bar{\eta}_1^{\ast} 
\end{array}\right), \qquad  \qquad
\widetilde{\mathbb{T}}_c = \left( 
\begin{array}{ccc}
\star & 0 & 0 \\
0 & \star & 0 \\
0 & 0 & 4\,\bar{\eta}_2 
\end{array}
\right), \nonumber
\end{align}
\normalsize
where we denoted by a star those components, which work on out-of-plane macro- and micro-strains and do not play any role in the considered plane-strain problem.

\section{Bulk wave propagation in Cauchy and relaxed micromorphic continua}\label{sec:wave_prop}
\subsection{Isotropic Cauchy continuum}
We make the plane-wave ansatz for the solution to \eqref{eq:Cauchystrong}:
\begin{equation}\label{eq:ansatzCauchy1}
u (x_1,x_2,t) = \widehat{\psi}\, e^{i(\langle k,x \rangle -\omega t)} =\widehat{\psi}\, e^{i(k_1 x_1+ k_2\, x_2 - \omega t)}, \quad \widehat{\psi} \in \C^2,
\end{equation}
where $k= (k_1,k_2)^{\rm T}$ is the wave vector and $\omega$ is the angular frequency.\footnote{As we will show in the following, $k_2$, which is the second component of the wave-number, is always supposed to be known and is given by Snell's law when imposing boundary conditions on a given boundary.} Plugging \eqref{eq:ansatzCauchy1} into equation \eqref{eq:Cauchystrong} we get a $2\times 2$ algebraic system of the form
\begin{equation}\label{eq:algebraic1}
A \cdot \widehat{\psi} = 0,  
\end{equation}
where 
\begin{equation}\label{eq:coeffsmatrixCauchy}
A = \left(\begin{array}{cc}
\rho\,\omega^2 - (2\mu + \lambda)\,k_1^2 - \mu\,k_2^2 & -(\mu + \lambda)~ k_1\,k_2 \\
-(\mu + \lambda)\,k_1\,k_2  &\rho\,\omega^2 - (2\mu + \lambda)\,k_2^2 - \mu\,k_1^2
\end{array}\right). 
\end{equation}
The algebraic system \eqref{eq:algebraic1} has a solution if and only if $\det A = 0$. This is a bi-quadratic polynomial equation which has the following four solutions:\small
\begin{equation}\label{eq:k1Cauchy}
k_1^{L,\rm r} = -\sqrt{\frac{\rho}{2\mu + \lambda}\omega^2-(k_2^{L,\rm r})^2}, \quad   k_1^{S,\rm r} = - \sqrt{\frac{\rho}{\mu}\omega^2 - (k_2^{S,\rm r})^2}, \quad k_1^{L,\rm t} = \sqrt{\frac{\rho}{2\mu + \lambda}\omega^2-(k_2^{L,\rm t})^2}, \quad   k_1^{S,\rm t} =  \sqrt{\frac{\rho}{\mu}\omega^2 - (k_2^{S,\rm t})^2},
\end{equation}\normalsize
where we denote by L and S the longitudinal and shear waves and by $\rm r$ and $\rm t$ whether they  travel towards $-\infty$ or $+\infty$, respectively.
We plug the solutions \eqref{eq:k1Cauchy} into \eqref{eq:algebraic1} to calculate the corresponding eigenvectors
\begin{equation}\label{eq:eigenvectorsCauchy}
\widehat{\psi}^{L, \rm r}=\left( \begin{array}{c}
1\\
-\frac{k_2^{L, \rm r}}{k_1^{L, \rm r}}
\end{array}\right),\widehat{\psi}^{S, \rm r}
=\left( 
\begin{array}{c}
1\\
\frac{k_1^{S, \rm r}}{k_2^{S, \rm r}}
\end{array}  \right), \quad 
\widehat{\psi}^{L, \rm t}=\left( \begin{array}{c}
1\\
\frac{k_2^{L, \rm t}}{k_1^{L, \rm t}}
\end{array}\right), \quad \widehat{\psi}^{S, \rm t}
=\left( 
\begin{array}{c}
1\\
-\frac{k_1^{S, \rm t}}{k_2^{S, \rm t}}
\end{array}  \right).
\end{equation}
Normalizing these eigenvectors gives:
\begin{equation}\label{eq:eigenvectorsCauchynormal}
\psi^{L, \rm r} := \frac{1}{\left|\widehat{\psi}^{L, \rm r}\right|}\widehat{\psi}^{L, \rm r}, \quad \psi^{S, \rm r} := \frac{1}{\left|\widehat{\psi}^{S, \rm r}\right|}\widehat{\psi}^{S, \rm r}, \quad
\psi^{L, \rm t} := \frac{1}{\left|\widehat{\psi}^{L, \rm t}\right|}\widehat{\psi}^{L, \rm t}, \quad \psi^{S, \rm t} := \frac{1}{\left|\widehat{\psi}^{S, \rm t}\right|}\widehat{\psi}^{S, \rm t}.
\end{equation}
Then, the general solution to \eqref{eq:Cauchystrong} can be written as:
\small
\begin{align}
u(x_1,x_2,t) = \,&a^{L, \rm r} \psi^{L, \rm r} e^{i \left(k_1^{L, \rm r} x_1 + k_2^{L, \rm r}\, x_2 - \omega t\right)} + a^{S, \rm r} \psi^{S, \rm r} e^{i \left( k_1^{S, \rm r} x_1 + k_2^{S, \rm r}\, x_2 - \omega t\right)}\nonumber \\
&+a^{L, \rm t} \psi^{L, \rm t} e^{i \left(k_1^{L, \rm t} x_1 + k_2^{L, \rm t}\, x_2 - \omega t\right)} + a^{S, \rm t} \psi^{S, \rm t} e^{i \left( k_1^{S, \rm t} x_1 + k_2^{S, \rm t}\, x_2 - \omega t\right)} \label{eq:planewavesolution},
\end{align}
\normalsize
where $a^{L, \rm r},a^{S, \rm r},a^{L, \rm t}, a^{S, \rm t} \in \C$ are constants to be determined from the boundary conditions, $k_1^{L, \rm t}=-k_1^{L, \rm r}$ and $k_1^{S, \rm t}=-k_1^{S, \rm r}$, according to \eqref{eq:k1Cauchy}. Depending on the specific problems which are considered (e.g. semi-infinite media), only some modes may propagate in specific directions. In this case, some of the terms in the sum \eqref{eq:planewavesolution} have to be omitted. In particular, if we are considering waves propagating in the $x_1<0$ half-space, then the solution to \eqref{eq:Cauchystrong} reduces to:\footnote{
This choice for the sign of $k_1$ always gives rise to a solution which verifies conservation of energy at the interface. This particular choice, which at a first instance is rather intuitive, is not always the correct one when dealing with a relaxed micromorphic medium. }
\small
\begin{equation}\label{eq:planewavesolutionminus}
u(x_1,x_2,t) = a^{L, \rm r} \psi^{L, \rm r} e^{i \left(k_1^{L, \rm r} x_1 + k_2^{L, \rm r}\, x_2 - \omega t\right)} + a^{S, \rm r} \psi^{S, \rm r} e^{i \left( k_1^{S, \rm r} x_1 + k_2^{S, \rm r}\, x_2 - \omega t\right)},
\end{equation}
\normalsize
while if we are considering the $x_1>0$ half-space, then the solution to \eqref{eq:Cauchystrong} reduces to:
\small
\begin{equation}\label{eq:planewavesolutionplus}
u(x_1,x_2,t) = a^{L, \rm t} \psi^{L, \rm t} e^{i \left(k_1^{L, \rm t} x_1 + k_2^{L, \rm t}\, x_2 - \omega t\right)} + a^{S, \rm t} \psi^{S, \rm t} e^{i \left( k_1^{S, \rm t} x_1 + k_2^{S, \rm t}\, x_2 - \omega t\right)}.
\end{equation}
\normalsize
\subsection{Relaxed micromorphic continuum}\label{sec:wavepropRMM}
We start by collecting the unknown fields for the plane-strain case in a new variable:
\begin{equation}\label{eq:RMMnewvariable}
	v: = (u_1,u_2,P_{11},P_{12},P_{21},P_{22})^{\rm T}.
\end{equation}
The plane-wave ansatz for this unknown field reads:
\begin{equation}\label{eq:RMMplanewave}
v = \widehat{\phi}\, e^{i\left(\langle k,x\rangle -\omega t\right)} = \widehat{\phi}\, e^{i\left(k_1\,x_1 +k_2 x_2-\omega\, t\right)},
\end{equation}
where $\widehat{\phi}\in \C^6$ is the vector of amplitudes, $k = (k_1,k_2)^{\rm T} \in \C^2$ is the wave-vector\footnote{Here again, as in the case of a Cauchy medium, $k_2$ will be fixed and given by Snell's law when imposing boundary conditions.} and $\omega$ is the angular frequency. We plug the wave-form \eqref{eq:RMMplanewave} into \eqref{eq:governingRMMclosed} and get an algebraic system of the form 
\begin{equation}\label{eq:algebraic2}
\widehat{A}(k_1,k_2,\omega)\cdot \widehat{\phi} = 0,
\end{equation}
where $\widehat{A}(k_1,k_2,\omega) \in \C^{6 \times 6}$ is a matrix depending on $k_1,k_2,\omega$ and all the material parameters of the plane-strain tetragonal relaxed micromorphic model (see \ref{app:matrixA} for an explicit presentation of this matrix). In order for this system to have a solution other than the trivial one, we impose $\det \widehat{A} = 0$. 

The equation  $\det \widehat{A} = 0$ is a polynomial of order $12$ in $\omega$ and it involves only even powers of $\omega$. This means that, plotting the roots $\omega = \omega(k)$ gives $6$ dispersion curves in the $\omega-k$ plane (see Fig. \ref{fig:DispersionCurves}). On the other hand, the same polynomial is of order $8$ (and again involves only even powers of $k_1$), if regarded as a polynomial of $k_1$ ($k_2$ is supposed to be known when imposing boundary conditions). We can write the roots of the characteristic polynomial as:
\small
\begin{align}\label{eq:solutionsDetRMM}
k_1^{\op}(k_2,\omega),\quad k_1^{\twop}(k_2,\omega),& \quad k_1^{\thp}(k_2,\omega),\quad  k_1^{\fp}(k_2,\omega),  \\
k_1^{\fip}(k_2,\omega) = -k_1^{\op}(k_2,\omega),\quad k_1^{\sixp}(k_2,\omega) =  -k_1^{\thp}(k_2,\omega), &\quad k_1^{\sevp}(k_2,\omega) = - k_1^{\thp}(k_2,\omega), \quad k_1^{\eip}(k_2,\omega) = - k_1^{\fp}(k_2,\omega). \nonumber
\end{align}
\normalsize

We plug the solutions \eqref{eq:solutionsDetRMM} into \eqref{eq:algebraic2} and calculate the eigenvectors of $\widehat{A}$, which we denote by: $
\widehat{\phi}^{\op},\quad \widehat{\phi}^{\twop}, \quad \widehat{\phi}^{\thp}, \quad \widehat{\phi}^{\fp}, \quad \widehat{\phi}^{\fip}, \quad \widehat{\phi}^{\sixp}, \quad \widehat{\phi}^{\sevp}, \quad \widehat{\phi}^{\eip}. 
$
We normalize these eigenvectors, thus introducing the normal vectors
\begin{align}\label{eq:eigenvectorsRMMnormal}
&\phi^{\op} := \frac{1}{|\widehat{\phi}^{\op}|}\widehat{\phi}^{\op}, \quad \phi^{\twop} := \frac{1}{|\widehat{\phi}^{\twop}|}\widehat{\phi}^{\twop}, \quad \phi^{\thp} := \frac{1}{|\widehat{\phi}^{\thp}|}\widehat{\phi}^{\thp}, \quad \phi^{\fp} := \frac{1}{|\widehat{\phi}^{\fp}|}\widehat{\phi}^{\fp}, \nonumber \\
&\phi^{\fip} := \frac{1}{|\widehat{\phi}^{\fip}|}\widehat{\phi}^{\fip}, \quad \phi^{\sixp} := \frac{1}{|\widehat{\phi}^{\sixp}|}\widehat{\phi}^{\sixp}, \quad \phi^{\sevp} := \frac{1}{|\widehat{\phi}^{\sevp}|}\widehat{\phi}^{\sevp}, \quad \phi^{\eip} := \frac{1}{|\widehat{\phi}^{\eip}|}\widehat{\phi}^{\eip}.
\end{align}
Considering a micromorphic medium in which all waves travel simultaneously, the solution to \eqref{eq:governingRMMclosed} is:

\begin{equation}\label{eq:planewavesolutionRMM2}
v(x_1,x_2,t) = \sum_{j=1}^{8} \alpha_j \phi^{(j)} e^{i \left(k_1^{(j)} x_1 + k_2^{(j)}\, x_2 - \omega t\right)},
\end{equation}  
\normalsize
where $\alpha_j \in \C$ are unknown constants which will be determined from the boundary conditions.  
If, on the basis of the particular interface problem one wants to study, only some waves travel in the considered medium, then the extra waves must be omitted from the sum in \eqref{eq:planewavesolutionRMM2}. This means that if we are considering waves traveling in the $x_1>0$ direction, then the solution to \eqref{eq:governingRMMclosed} is given by \small
\begin{align}\label{eq:planewavesolutionRMM2plus}
v(x_1,x_2,t) =& \, \alpha_1\, \phi^{(1)} e^{i \left(k_1^{(1)} x_1 + k_2^{(1)}\, x_2 - \omega t\right)} + \alpha_2\,\phi^{(2)} e^{i \left(k_1^{(2)} x_1 + k_2^{(2)}\, x_2 - \omega t\right)} + \alpha_3\, \phi^{(3)} e^{i \left(k_1^{(3)} x_1 + k_2^{(3)}\, x_2 - \omega t\right)}\nonumber \\
&+ \alpha_4\,\phi^{(4)} e^{i \left(k_1^{(4)} x_1 + k_2^{(4)}\, x_2 - \omega t\right)},
\end{align}  \normalsize
while if we consider waves only traveling in the $x_1<0$ direction, the solution to \eqref{eq:governingRMMclosed} is given by \small
\begin{align}\label{eq:planewavesolutionRMM2minus}
v(x_1,x_2,t) =& \, \alpha_5\, \phi^{(5)} e^{i \left(k_1^{(5)} x_1 + k_2^{(5)}\, x_2 - \omega t\right)} + \alpha_6\,\phi^{(6)} e^{i \left(k_1^{(6)} x_1 + k_2^{(6)}\, x_2 - \omega t\right)} + \alpha_7\, \phi^{(7)} e^{i \left(k_1^{(7)} x_1 + k_2^{(7)}\, x_2 - \omega t\right)}\nonumber \\
&+ \alpha_8\,\phi^{(8)} e^{i \left(k_1^{(8)} x_1 + k_2^{(8)}\, x_2 - \omega t\right)}.
\end{align}   \normalsize

\section{Boundary Conditions}\label{sec:BCs}
We will consider two types of interface problems: (i) a \textbf{single interface} separating a Cauchy and a relaxed micromorphic medium, both assumed to be semi-infinite and (ii) a \textbf{micromorphic slab} of finite size embedded between two semi-infinite Cauchy media. In the following, we will simply denote ``single interface'' and ``micromorphic slab'' the first and second problem, respectively. 

In the single interface problem, two homogeneous infinite half-spaces are occupied by two materials in perfect contact with each other. The material on the left of the interface is an isotropic classical Cauchy medium, while the material on the right is a microstructured tetragonal metamaterial modeled by the tetragonal relaxed micromorphic model (see Fig \ref{fig:singleinterface}). 

In the micromorphic slab problem, two homgeneous infinite half-spaces are separated by a micromorphic slab of finite width $h$. Three materials are thus in perfect contact with each other: the material on the left of the first interface is a classical linear elastic isotropic Cauchy medium, the material in the middle is an anisotropic relaxed micromorphic medium, while the material on the right of the second interface is again a classical isotropic Cauchy medium (see Fig. \ref{fig:slab}). 
\vspace{0.51cm}
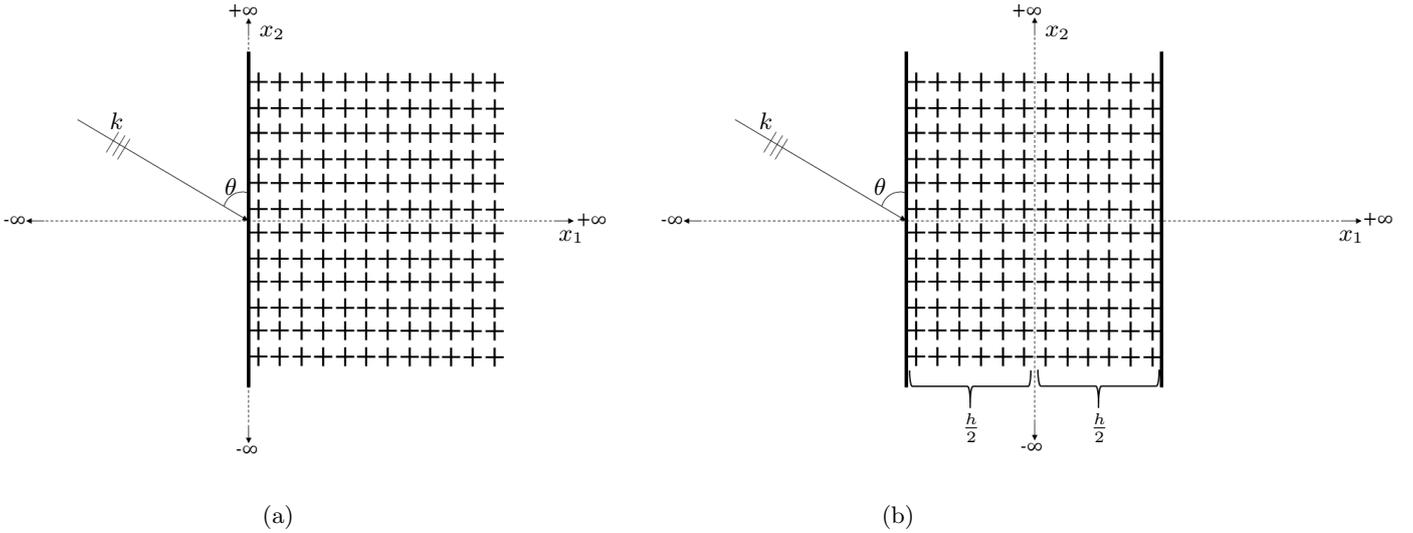
\begin{figure}[h]
	\hspace{-3cm}
	\begin{subfigure}{0.66\linewidth}
		\centering
		\includegraphics[trim={1cm 2cm  0cm  0cm},clip,
		scale=0.35]{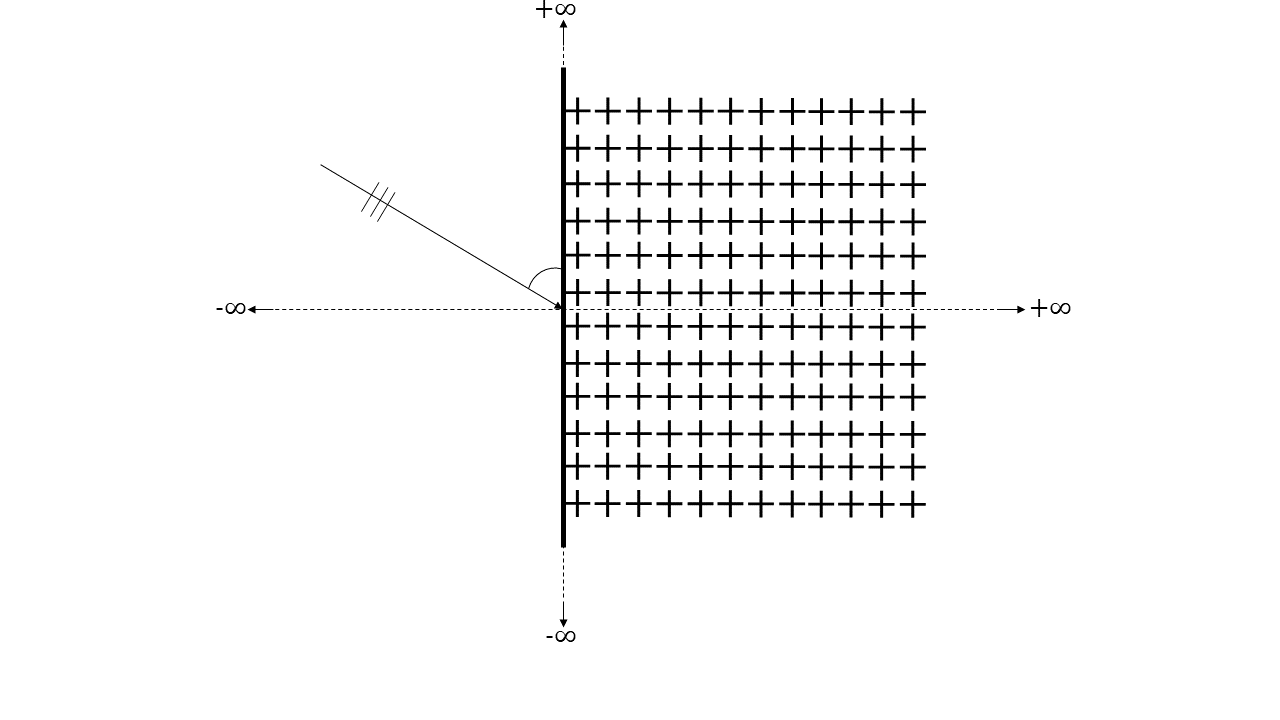}
		\begin{picture}(0,0)(0,0)
		\put(-20,110) {\small$\theta$}
		\put(-63,135) {\small$k$}
		\put(-7,170) {\small$x_2$}
		\put(105,93) {\small$x_1$}
		\end{picture}
		\caption{}
		\label{fig:singleinterface}
	\end{subfigure}
	\begin{subfigure}{0.35\linewidth}
		\centering
		\hspace*{-1cm}
		\includegraphics[
		trim={4cm 2cm  0cm  0cm},clip,
		scale=0.35]{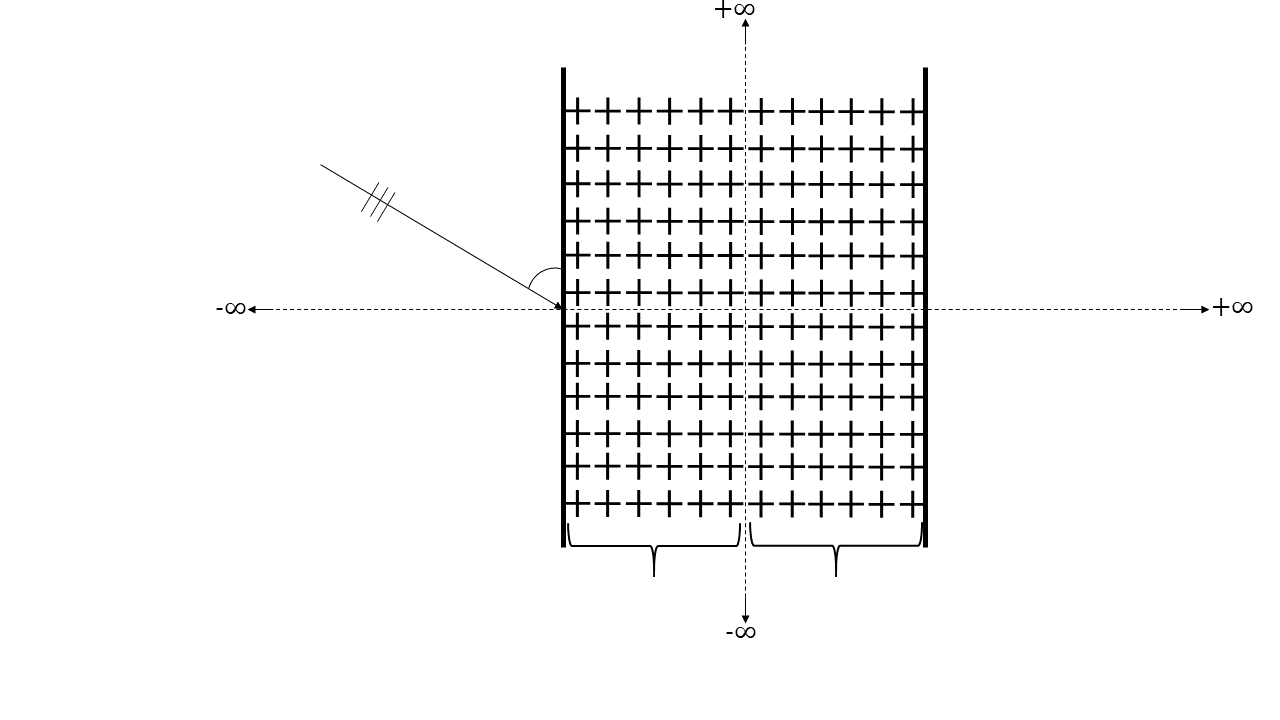}\\	
		\begin{picture}(0,0)(0,0)
		\put(-9,110) {\small$\theta$}
		\put(-52,135) {\small$k$}
		\put(55,170) {\small$x_2$}
		\put(165,93) {\small$x_1$}
		\put(24,20) {\small$\frac{h}{2}$}
		\put(72,20) {\small$\frac{h}{2}$}
		\end{picture}
		\caption{}
		\label{fig:slab}
	\end{subfigure}\hfill
	\caption{\label{fig:interfaces}  Panel (a): single interface separating a Cauchy medium from a relaxed micromorphic medium (both semi-infinite in the $x_1$ direction). Panel (b): A micromorphic slab of width $h$ between two semi-infinite elastic Cauchy media. Both configurations (a) and (b) are semi-infinite in the $x_2$ direction.}
\end{figure}

\subsection{Boundary conditions at an interface between a Cauchy continuum and a relaxed micromorphic continuum}\label{sec:BCsingle}
As is well established in previous works \cite{madeo2016reflection,aivaliotis2018low,aivaliotis2019microstructure} there are four boundary conditions which can be imposed at a Cauchy/relaxed micromorphic interface. These are continuity of displacement, continuity of generalized traction and an extra set of boundary conditions which involve only the micro-distortion tensor $P$. These additional boundary conditions stem from the least action principle, and are called \emph{free} and \emph{fixed microstructure} boundary conditions \cite{aivaliotis2019microstructure}, depending on whether we assign the dual quantity to the micro-distortion, or the micro-distortion itself. 

For the displacement, we have:
\begin{equation}\label{eq:displcont}
[[u]] =0 \Rightarrow u^{-} = u^{+}, \qquad \text{on} \quad x_1 = 0,
\end{equation}
where $u^{-}$ is the macroscopic displacement on the ``minus'' side (the $x_1<0$ half-plane, occupied by an isotropic Cauchy medium) and $u^{+}$ is the macroscopic displacement on the ``plus'' side (the $x_1>0$ half-plane, occupied by an anisotropic relaxed micromorphic medium: the first two components of the vector $v$, defined in \eqref{eq:RMMnewvariable}). As for the jump of generalized traction we have:
\begin{equation}\label{jumpforcevec}
t = \widetilde{t},
\end{equation}
where $t$ is the Cauchy traction on the ``minus'' side and $\widetilde{t}$ is the generalized traction on the ``plus'' side. We recall that in a Cauchy medium, $t = \sigma \cdot \nu$, $\nu$ being the outward unit normal to the surface and $\sigma$ being the Cauchy stress tensor given by \eqref{eq:Cauchystress}. The generalized traction for the relaxed micromorphic medium is given by
\begin{equation}\label{eq:forceRMM}
\widetilde{t} = \left(\widetilde{\sigma} + \widehat{\sigma}\right) \cdot \nu, \qquad \widetilde{t}_{i} = \left(\widetilde{\sigma}_{ij} + \widehat{\sigma}_{ij}\right)\cdot \nu_j, \qquad \text{on} \quad x_1 = 0,
\end{equation}
where $\widetilde{\sigma}, \widehat{\sigma}$ are defined in \eqref{eq:relaxedRHS2}.

In order to impose suitable boundary conditions for the micro-distortion tensor $P$, we need the concept of the double traction $\tau$ which is the dual quantity of the micro-distortion tensor $P$ and is defined as \cite{madeo2016reflection,aivaliotis2019microstructure}:
\begin{equation}\label{eq:doubleforcedef}
\tau=-m\cdot \epsilon \cdot \nu, \qquad \tau_{ij} = -m_{ik}\,\epsilon_{kjh}\,\nu_{h}
\end{equation}
where $m$ is given in \eqref{eq:relaxedRHS2}, $\epsilon$ is the Levi-Civita tensor and $\nu$ is the outward unit normal to the interface.

\subsubsection{Free microstructure}
In this case, the macroscopic displacement and generalized tractions are continuous while the microstructure of the medium is free to move along the interface \cite{madeo2016first,madeo2016reflection,madeo2017relaxed,aivaliotis2019microstructure}. Leaving the interface free to move means that $P$ is arbitrary, which, on the other hand, implies that the double traction $\tau$ must vanish. We have then:
\begin{equation}\label{eq:BCfreemicro}
[[u_i]]=0, \quad \textcolor{black}{\widetilde{t}}_i -\textcolor{black}{t}_i = 0,\quad \tau_{ij} = 0, \quad i=1,2, \text{ } j=2.
\end{equation}

\subsubsection{Fixed microstructure}
This is the case in which we impose that the microstructure on the relaxed micromorphic side does not vibrate at the interface. The boundary conditions in this case are \cite{madeo2016reflection,aivaliotis2019microstructure}:\footnote{We remark that, in the relaxed micromorphic model, only the tangent part of the double traction in \eqref{eq:BCfreemicro} or of the micro-distortion tensor in \eqref{eq:BCfixedmicro} must be assigned, see \cite{neff2015relaxed,neff2014unifying} for more details.}
\begin{equation}\label{eq:BCfixedmicro}
[[u_i]]=0,\quad \textcolor{black}{\widetilde{t}}_i -\textcolor{black}{t}_i=0,\quad P_{ij} = 0,\quad i=1,2, \text{ } j=2.
\end{equation}

\subsubsection{Continuity of macroscopic displacement and of generalized traction implies conservation of energy at the interface}
We have previously shown that conservation of energy for a bulk Cauchy and relaxed micromorphic medium is given by equation \eqref{eq:EnergyConservation}, where the energy flux is defined in \eqref{eq:Cauchyflux} and \eqref{eq:fluxAniso}, respectively. It is important to remark that the conservation of energy \eqref{eq:EnergyConservation} has a ``boundary counterpart''. This establishes that the jump of the normal part of the flux must be vanishing, or, in other words, the normal part of the flux must be continuous at the considered interface (this comes from the bulk conservation law and the use of the Gauss divergence theorem). In symbols, when considering a surface $\Sigma$ separating two continuous media, we have 
\begin{equation}\label{eq:jumpflux}
	[[H \cdot \nu]] = 0, \quad \text{on} \quad \Sigma.
\end{equation}
We want to focus the reader's attention on the fact that, in the framework of a consistent theory in which the bulk equations and boundary conditions are simultaneously derived by means of a variational principle, the jump conditions imposed on $\Sigma$ necessarily imply the surface conservation of energy \eqref{eq:jumpflux}, as far as a conservative system is considered. We explicitly show here that this is true for an interface $\Sigma$ separating a Cauchy medium from a relaxed micromorphic one. The same arguments, however, hold for interfaces between two Cauchy or two relaxed micromorphic media. 

To that end, considering for simplicity that the interface $\Sigma$ is located at $x_1=0$ (so that its normal is $\nu = (1,0)^{\rm T}$) 
 we get from equation \eqref{eq:fluxAniso} that the normal flux computed on the ``relaxed micromorphic'' side is given by:

\begin{equation}\label{eq:fluxAniso1}
	(H\cdot \nu)^{+} := H_1^{+} =-\nu^{+}\cdot \left( (\widetilde{\sigma} + \widehat{\sigma})^{\rm T} \cdot u_{,t}^{+}  -\left( m^T \cdot P_{,t} \right):\epsilon \right) \quad \text{at} \quad x_1=0.
\end{equation} 

By the same reasoning, the flux at the interface on the ``Cauchy'' side is computed from \eqref{eq:Cauchyflux} and gives:
\begin{equation}\label{eq:Cauchyflux1}
	(H\cdot \nu)^{-} := H_1^{-} = -\nu^{-} \cdot \sigma \cdot u_{,t}^{-}.
\end{equation}
Equation \eqref{eq:jumpflux} can then be rewritten as:
\begin{equation}\label{eq:jumpflux2}
-\nu^{+}\cdot \left( (\widetilde{\sigma} + \widehat{\sigma})^{\rm T} \cdot u_{,t}^{+}  -\left( m^T \cdot P_{,t} \right):\epsilon \right) = -\nu^{-} \cdot \sigma \cdot u_{,t}^{-}.
\end{equation}
It is clear that, given the jump conditions \eqref{eq:BCfreemicro} or \eqref{eq:BCfixedmicro} the latter relation is automatically verified.

As we will show in the remainder of this paper, when modeling a metamaterial's boundary via the relaxed micromorphic model, we only need a finite number of modes in order to exactly verify boundary conditions and, consequently, surface energy conservation. This provides one of the most powerful simplifications of the relaxed micromorphic model with respect to classical homogenization methods, in which infinite modes are needed to satisfy conservation of stress and displacement at the metamaterial's boundary (see \cite{willis2011effective}).

\subsection{Boundary conditions for a micromorphic slab embedded between two Cauchy media}\label{sec:BCsslab}
The boundary conditions to be satisfied at the two interfaces separating the slab from the two Cauchy media in Fig. \ref{fig:slab}, are continuity of displacement, continuity of generalized traction and a condition on the micro-distortion tensor $P$ (free or fixed microstructure). This means that we have twelve sets of scalar boundary conditions, six on each interface. The finite slab has width $h$ and we assume that the two interfaces are positioned at $x_1=-h/2$ and $x_1=h/2$, respectively. The continuity of displacement conditions to be satisfied at the two interfaces of the slab are:\footnote{We denote by $\widetilde{v}$ the first two components of the micromorphic field $v$ defined in equation \eqref{eq:RMMplanewave}.}
\begin{equation}\label{eq:jumpdisplslab}
u^{-}_i = \widetilde{v}_i, \text{ on } x_1=-\frac{h}{2},   \qquad
\widetilde{v}_i = u^{+}_i,  \text{ on } x_1=\frac{h}{2}, \qquad i = 1,2.
\end{equation}
As for the continuity of generalized traction, we have:
\begin{equation}\label{eq:jumptractionslab}
t^{-}_i = \widetilde{t}_i,  \text{ on } x_1=-\frac{h}{2}, \qquad
\widetilde{t}_i= t^{+}_i,   \text{ on } x_1=\frac{h}{2} \qquad i = 1,2,
\end{equation}
where $t^{\pm} = \sigma^{\pm}\cdot\nu^{\pm}$ are classical Cauchy tractions and $\widetilde{t}$ is again given by \eqref{eq:forceRMM}.

Furthermore, we need the additional free or fixed microstructure boundary conditions to be satisfied at both interfaces. These are given by: 
\begin{equation}\label{eq:freeslab}
\tau_{ij} = 0, \text{ on } x_1=-\frac{h}{2},   \qquad
\tau_{ij} = 0,  \text{ on } x_1=\frac{h}{2}, \qquad i=1,2,\,\, j = 2,
\end{equation}
or 
\begin{equation}\label{eq:fixedslab}
P_{ij} = 0, \text{ on } x_1=-\frac{h}{2},   \qquad
P_{ij} = 0,  \text{ on } x_1=\frac{h}{2}, \qquad i=1,2,\,\, j = 2,
\end{equation}
were the expression for the double traction $\tau$ is given in \eqref{eq:doubleforcedef}. We remark once again, that only the tangential parts of the double traction $\tau$ or the micro-distortion tensor $P$ must be assigned in order to fulfill conservation of energy. 

\section{Reflection and transmission at the single Cauchy/relaxed micromorphic interface}\label{sec:RTSingle}

In this section, we study the two-dimensional, plane-strain, time-harmonic scattering problem from an anisotropic micromorphic half-space (see equations \eqref{eq:governingRMMclosed}), occupying the region $x_1>0$ of Fig. \ref{fig:singleinterface}. With reference to Fig. \ref{fig:singleinterface}, the half-space $x_1<0$ is filled with a linear elastic Cauchy continuum, governed by  equation \eqref{eq:Cauchystrong}. Considering that reflected waves only travel in the $x_1<0$ Cauchy half-plane, only negative solutions for the $k_1$'s must be kept in equation \eqref{eq:k1Cauchy}, so that the total solution in the left half-space is 
\small
\begin{align}\label{eq:solplusSingle}
u^{-}(x_1,x_2,t) &= a^{L/S,\mathrm{i}} \psi^{L/S,\mathrm{i}} e^{i\left(\left\langle x, k^{L/S, \rm i}\right\rangle -\omega t\right)} + a^{L, \rm r} \psi^{L,\rm r} e^{i\left(\left\langle x, k^{L,\rm r}\right\rangle -\omega t\right)} + a^{S,\rm r} \psi^{S,\rm r} e^{i\left(\left\langle  x,k^{S,\rm r}\right\rangle -\omega t\right)} \nonumber \\
&=: u^{L/S,\rm i}+ u^{L,\rm r} + u^{S,\rm r},
\end{align} 
\normalsize
where we write L or S in the incident wave depending on whether the wave is longitudinal or shear and $\rm i$ and $\rm r$ in the exponents stand for ``incident'' and ``reflected''. Analogously, the solution on the right half-space, which is occupied by a relaxed micromorphic medium, is\footnote{We suppose here that, for the transmitted wave we have to consider only the positive solutions of \eqref{eq:solutionsDetRMM}. 
}
\small
\begin{align}\label{eq:solmicroSingle}
v(x_1,x_2,t) =& \, \alpha_1\, \phi^{(1)} e^{i \left(k_1^{(1)} x_1 + k_2^{(1)}\, x_2 - \omega t\right)} + \alpha_2\,\phi^{(2)} e^{i \left(k_1^{(2)} x_1 + k_2^{(2)}\, x_2 - \omega t\right)} + \alpha_3\, \phi^{(3)} e^{i \left(k_1^{(3)} x_1 + k_2^{(3)}\, x_2 - \omega t\right)}\nonumber \\
&+ \alpha_4\,\phi^{(4)} e^{i \left(k_1^{(4)} x_1 + k_2^{(4)}\, x_2 - \omega t\right)} + \alpha_5\,\phi^{(5)} e^{i \left(k_1^{(5)} x_1 + k_2^{(5)}\, x_2 - \omega t\right)},
\end{align} 
\normalsize
where 
we have kept only terms with positive $k_1$'s in the solution \eqref{eq:solutionsDetRMM}, since transmitted waves are supposed to propagate in the $x_1>0$ half-plane.
 
Since the incident wave is always propagative, the polarization and wave-vectors are given by: 
\begin{align}
\psi^{L,\mathrm{i}} &= (\sin \theta^{L}, -\cos \theta^{L})^{\rm T}, \quad k^{L} =|k^{L}| (\sin \theta^{L}, -\cos \theta^{L})^{\rm T}, \label{eq:incidentL} \\
\psi^{S,\mathrm{i}} &= (\cos \theta^S, \sin \theta^S)^{\rm T}, \quad\hspace{0.4cm} k^{S} = |k^S|(\sin \theta^S, -\cos \theta^S)^{\rm T}, \label{eq:incidentS}
\end{align}
where, according to \eqref{eq:k1Cauchy},
$|k^L|=\frac{\omega}{c_L}$ and $ |k^S|= \frac{\omega}{c_S}$,
with $c_L=\sqrt{(2\mu + \lambda)/\rho}$ and $c_S=\sqrt{\mu/\rho}$ the longitudinal and shear speeds of propagation and $ \theta^L$ and $ \theta^S$ the angles of incidence when the wave is longitudinal or shear, respectively (see Fig. \ref{fig:interfaces} and \cite{aivaliotis2019microstructure} for a more detailed exposition).

The continuity of displacement condition \eqref{eq:displcont} provides us with the generalized Snell's law for the case of a Cauchy/relaxed micromorphic interface (see \cite{aivaliotis2018low} for a detailed derivation):
\begin{empheq}[box=\fbox]{gather}
k_2^{L/S, \rm i} = k_2^{L,\rm r} = k_2^{S,\rm r} = k_2^{\op} = k_2^{\twop} = k_2^{\thp} = k_2^{\fp} = k_2^{\fip}. \label{eq:Snellmicromorphic}
\\\notag\text{\textbf{Generalized Snell's Law}}
\end{empheq}

As for the flux, the unit outward normal vector to the surface (the $x_2$ axis) is $\nu = (-1,0)$. This means that in expressions \eqref{eq:Cauchyflux} and \eqref{eq:fluxAniso} for the fluxes, we need only take into account the first component. According to our definitions \eqref{eq:Cauchyflux} and \eqref{eq:fluxAniso}, we have:
\begin{equation}\label{eq:fluxessingle}
H_1^{-}=-u_{i,t}\,\sigma_{i1}, \qquad
H_1^{+}=-v_{i,t}\left(\widetilde{\sigma}_{i1}+\widehat{\sigma}_{i1}\right) - m_{ih} P_{ij,t} \epsilon_{jh1}, \quad i=1,2, \,\, j,k=2,3.
\end{equation}
Having calculated the ``transmitted'' flux, we can now look at the reflection and transmission coefficients for the case of a Cauchy/relaxed micromorphic interface. We define\footnote{In order to  easily compute these coefficients in the numerical implementation of the code, we employ Lemma \ref{Lemma1} given in Appendix \ref{app:lemma}.}
\begin{equation}\label{eq:Js}
J^{\rm i} = \frac{1}{T}\int_0^{T} H^{\rm i}(x,t) dt, \quad J^{\rm r} = \frac{1}{T}\int_0^{T} H^{\rm r}(x,t) dt, \quad J^{\rm t} = \frac{1}{T}\int_0^{T} H^{\rm t}(x,t) dt,
\end{equation}
where $T$ is the time period of the considered harmonic waves, $H^{\rm i} =H_1^{-}(u^{L/S,\rm i})$, $H^{\rm r} = H_1^{-}(u^{L,\rm r} + u^{S,\rm r})$ and $H^{\rm t} = H_1^{+}\left(v\right)$, with $H_1^{+}$ and $H_1^{-}$ defined in \eqref{eq:fluxessingle}.

Then the reflection and transmission coefficients are 
\begin{align}\label{eq:reflcoeffsingle}
\mathcal{R}=\frac{J^{\rm r}}{J^{\rm i}}, \quad \mathcal{T}=\frac{J^{\rm t}}{J^{\rm i}}.
\end{align}

Since the system is conservative, we must have that $\mathcal{R} + \mathcal{T} = 1$.

\section{Reflection and transmission at a relaxed micromorphic slab}\label{sec:RTSlab}
As pointed out, in this case there are three media: the first Cauchy half-space, the anisotropic relaxed micromorphic slab and the second Cauchy half-space. The two Cauchy half-spaces are denoted by $-$ and $+$, while the quantities considered in the slab have their own notation. 

The solution on the first Cauchy half space is given, as in the case of a single interface, by
\begin{equation}\label{eq:solminusSlab}
u^{-}(x_1,x_2,t) = a^{L/S,\rm i} \psi^{L/S,\rm i} e^{i\left(\left\langle x, k^{L/S,\rm i}\right\rangle  -\omega t\right)} + a^{L,\rm r} \psi^{L,,\rm r} e^{i\left(\left\langle x, k^{L,,\rm r}\right\rangle -\omega t\right)} + a^{S,\rm r} \psi^{S,\rm r} e^{i\left(\left\langle  x,k^{S,\rm r}\right\rangle -\omega t\right)}.
\end{equation} 
\normalsize

In the case of a relaxed micromorphic slab, when solving the eigenvalue problem we must select and keep all the roots for $k_1$, as given in \eqref{eq:solutionsDetRMM}, both positive and negative. This is due to the fact that there are waves which transmit in the micromorphic part from the first interface ($x_1=-h/2$), upon which the incident wave hits and waves which reflect on the second interface ($x_1=h/2$). This means that the solution of the PDEs in the slab is \small
\begin{align}\label{eq:solmicroSlab} 
v(x_1,x_2,t) =\,& \alpha_1\, \phi^{\op} e^{i\left(\left\langle x,k^{\op} \right\rangle-\omega t\right)} + \alpha_2\, \phi^{\twop} e^{i\left(\left\langle x,k^{\twop} \right\rangle -\omega t\right)}+ \alpha_3\, \phi^{\thp} e^{i\left(\left\langle x,k^{\thp} \right\rangle-\omega t\right)} + \alpha_4\, \phi^{\fp} e^{i\left(\left\langle x,k^{\fp} \right\rangle-\omega t\right)}\nonumber \\
 &+ \alpha_5\, \phi^{\fip} e^{i\left(\left\langle x,k^{\fip} \right\rangle-\omega t\right)}
+ \alpha_6\, \phi^{\sixp} e^{i\left(\left\langle x,k^{\sixp} \right\rangle-\omega t\right)} + \alpha_7\, \phi^{\sevp} e^{i\left(\left\langle x,k^{\sevp} \right\rangle -\omega t\right)}+ \alpha_8\, \phi^{\eip} e^{i\left(\left\langle x,k^{\eip} \right\rangle-\omega t\right)}.
\end{align}\normalsize

Finally, the solution on the right Cauchy half-space is 
\begin{equation}\label{eq:solplusSlab}
u^{+}(x_1,x_2,t) = a^{L,\rm t} \psi^{L,\rm t} e^{i\left(\left\langle x, k^{L,\rm t}\right\rangle -\omega t\right)} + a^{S,\rm t} \psi^{S,\rm t} e^{i\left(\left\langle  x,k^{S,\rm t}\right\rangle-\omega t\right)}.
\end{equation}

The continuity of displacement conditions \eqref{eq:jumpdisplslab} again imply a generalized form of Snell's law for the case of the micromorphic slab (see \cite{aivaliotis2019microstructure}):
\begin{empheq}[box=\fbox]{gather}
k_2^{L/S,i} = k_2^{L,\rm r} = k_2^{S,\rm r} = k_2^{\op} = k_2^{\twop} = k_2^{\thp} = k_2^{\fp} = k_2^{\fip} = k_2^{\sixp} = k_2^{\sevp} = k_2^{\eip} = k_2^{L,\rm t} = k_2^{S,\rm t} \label{eq:SnellSlab}.
\\\notag\text{\textbf{Generalized Snell's law in a micromorphic slab}}
\end{empheq}

In order to define the reflection and transmission coefficients in the case of the anisotropic slab, we follow the same reasoning as for the single interface. However, in this case, the transmitted flux is defined as the flux on the right of the second interface, which is occupied by an isotropic Cauchy medium. 

In this case, the reflected flux is evaluated at $x_1 = -h/2$ and the transmitted flux at $x_1 = h/2$. Both the reflected and the transmitted fields propagate in isotropic Cauchy media, so that the quantities $J^{\rm i}, J^{\rm r},J^{\rm t}$ defined in \eqref{eq:Js} are given by (see \cite{aivaliotis2018low})
\footnotesize
\begin{align}
J^{\rm i}_{\rm slab} &= \frac{1}{2} \Re\Bigl(\Bigl[(2\mu^{-} + \lambda^{-}) \left|\psi_1^{L/S,\rm i}\right|^2 k_1^{L/S,\rm i} + \lambda^{-} \left(\psi_1^{L/S,\rm i}\right)^{*}\psi_2^{L/S,\rm i} k_2^{L/S,\rm i} \nonumber \\ 
&\hspace{4.95cm} +  \mu^{-} \left(\psi_1^{\rm i}\left(\psi_2^{L/S,\rm i}\right)^{*}k_2^{L/S,\rm i} + \left|\psi_2^{L/S,\rm i}\right|^2k_1^{L/S,\rm i}\right)\Bigr]\left|a^{L/S,\rm i}\right|^2\omega\Bigr),\label{eq:fluxISlab}\\
J^{\rm r}_{\rm slab} &= \sum_{j \in \mathcal{J}}	\frac{1}{2} \Re\left(\left[(2\mu^{-} + \lambda^{-}) \left|\psi_1^{j,\rm r}\right|^2 k_1^{j, \rm r} + \lambda^{-} \left(\psi_1^{j,\rm r}\right)^{*}\psi_2^{j,\rm r} k_2^{j,\rm i}
+ \mu^{-} \left(\psi_1^{j,\rm r}\left(\psi_2^{j,\twop}\right)^{*}k_2^{j,\rm i} \left|\psi_2^{j,\rm r}\right|^2k_1^{j,\rm r}\right)\right]\left|a^{j,\rm r}\right|^2\omega
\right), \label{eq:fluxRSlab}\\
J^{\rm t}_{\rm slab} &= \sum_{j \in \mathcal{J}}	\frac{1}{2} \Re\left(\left[(2\mu^{+} + \lambda^{+}) \left|\psi_1^{j,\rm t}\right|^2 k_1^{j,\rm t} + \lambda^{+} \left(\psi_1^{j,\rm t}\right)^{*}\psi_2^{j,\rm t} k_2^{j,\rm i} + \mu^{+} \left(\psi_1^{j,\rm t}\left(\psi_2^{j,\rm t}\right)^{*}k_2^{j,\rm i} + \left|\psi_2^{j,\rm t}\right|^2k_1^{j,\rm t}\right)\right]\left|a^{j,\rm t}\right|^2\omega\right),\hspace{-10pt}\label{eq:fluxTSlab} 
\end{align}
\normalsize
where, $\mathcal{J} = \{L,S\} $, $a^{j,\rm i}, a^{j,\rm r}, a^{j,\rm t} \in \C$ and $\psi^{j,\rm i}, \psi^{j,\twop}, \psi^{j, \rm t}\in \C^2$  with $j \in \mathcal{J}$, are the amplitudes and polarization vectors for incident, reflected and transmitted waves, respectively (see also equation \eqref{eq:eigenvectorsCauchynormal}), $\mu^{-},\lambda^{-}$ and $\mu^{+},\lambda^{+}$ are the Lam\'e parameters of the left and right Cauchy half-spaces, respectively.

The reflection and transmission coefficients for the slab can then be defined as:
\begin{align}\label{eq:reflcoeffslab}
\mathcal{R}_{\rm slab}=\frac{J^{\rm r}_{\rm slab}}{J^{\rm i}_{\rm slab}}, \qquad \mathcal{T}_{\rm slab}=\frac{J^{\rm t}_{\rm slab}}{J^{\rm i}_{\rm slab}}.
\end{align}
Since the system is conservative, we must have $\mathcal{R}_{\rm slab}+\mathcal{T}_{\rm slab} = 1$.

\section{Reflective properties of the detailed micro-structured slab}\label{sec:Comsol}
In this section, we consider the scattering of in-plane elastic waves from a slab containing cross like holes drilled in an isotropic linear elastic material (see Fig. \ref{fig:geom_scatt}(a)). The holes in the micro-structured slab are arranged according to a truncated square lattice, \emph{i.e.} a finite number $N$ of cells in the $x_1$ direction and an infinite number of cells in the $x_2$ direction. The details of the unit cell are given in Fig. \ref{fig:geom_scatt2} and Tab. \ref{tab:microstructureparams}.
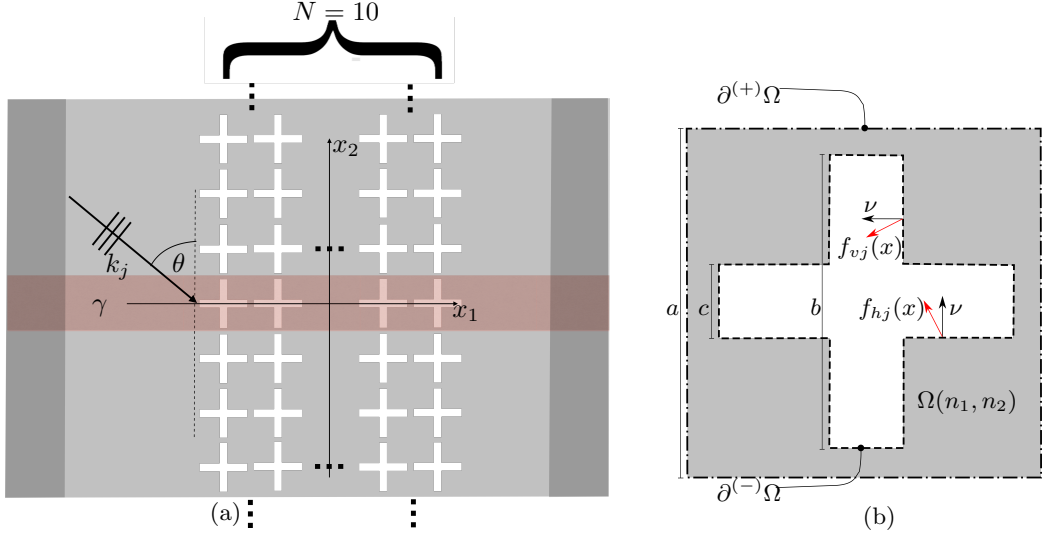
\begin{figure}[h!]
	\centering
	\begin{subfigure}{0.5\linewidth}
		\centering
		\hspace*{1cm}
		\vspace{-1cm}
		\includegraphics[
		width=\textwidth,angle=0]{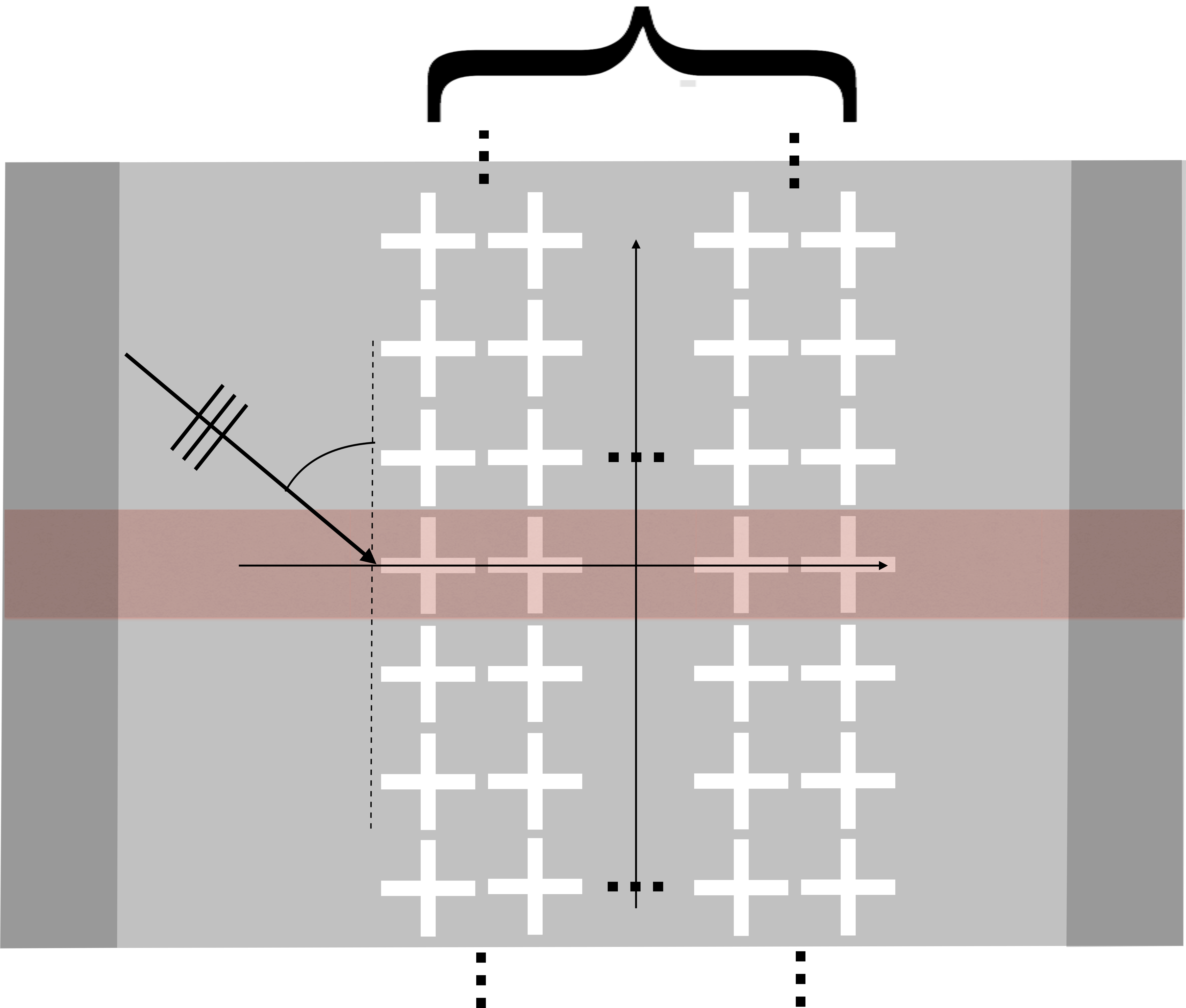}
		\begin{picture}(0,0)(0,0)
		\put(-20,80) {$\theta	$}
		\put(-45,80) {$ k_j$}
		\put(40,125) {$x_2$}
		\put(85,63) {$x_1$}
		\put(25,175) {$N=10$}
		\put(-50,65) {$\gamma$}
		\end{picture}
		\caption{}
		\label{fig:geom_scatt1}
	\end{subfigure}
	\hspace{2cm}
	\begin{subfigure}{0.3\linewidth}
		\centering
		\vspace{1cm}
		\hspace*{-0.5cm}
		\includegraphics[
		width=\textwidth]{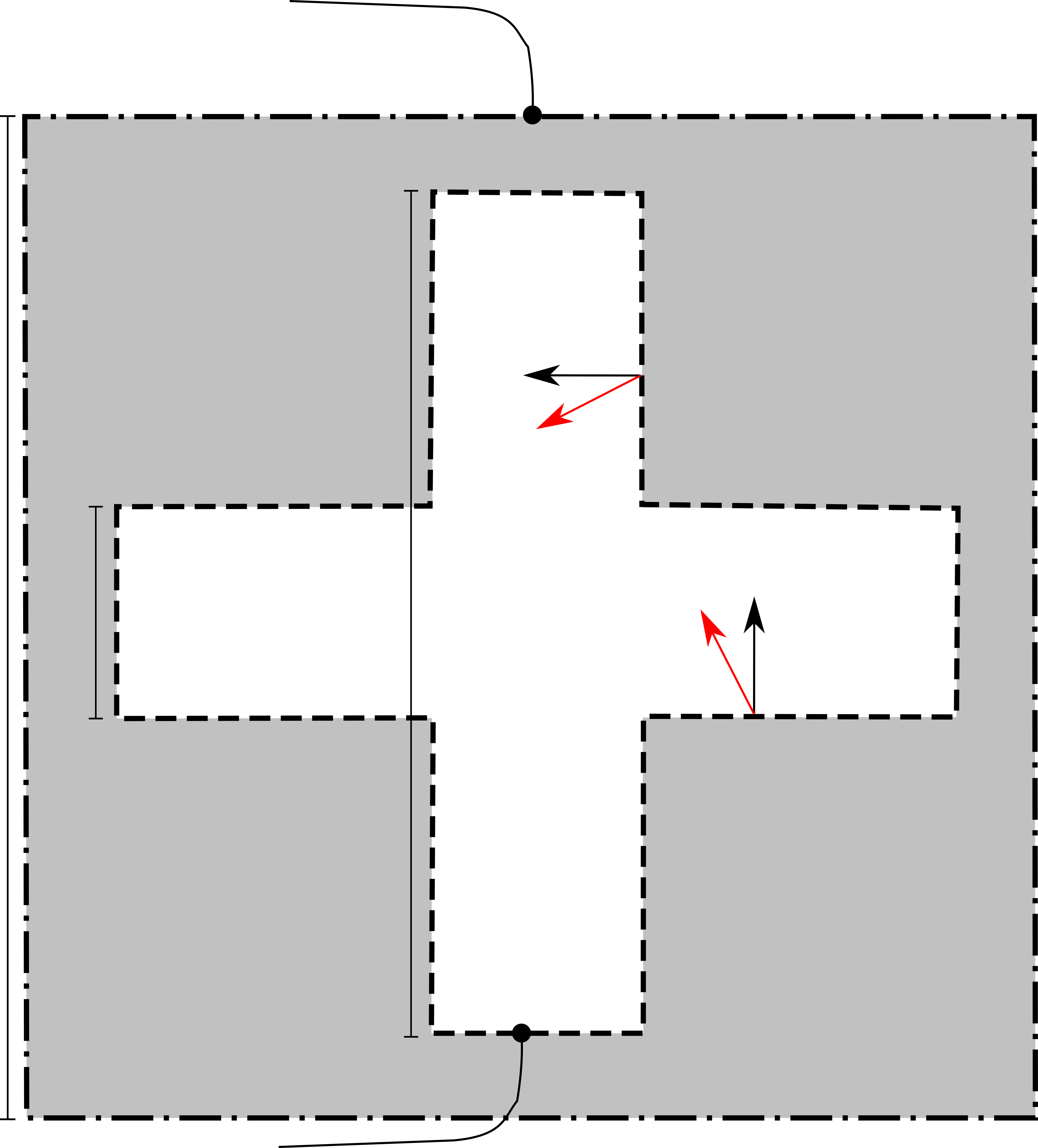}
		\begin{picture}(0,0)(0,0)
		\begin{small}
		\put(-38,65) {$\nu$}
		\put(-72,65) {$f_{hj}(x)$}
		\put(-125,145) {$\partial^{(+)}\Omega$}
		\put(-125,-5) {$\partial^{(-)}\Omega$}
		\put(-50,30) {$\Omega(n_1,n_2)$}
		\put(-144,65) {$a$}
		\put(-132,65) {$c$}
		\put(-90,65) {$b$}
		\put(-70,104) {$\nu$}
		\put(-80,88) {$f_{vj}(x)$}
		\end{small}
		\end{picture}
		\caption{}
		\label{fig:geom_scatt2}	
	\end{subfigure}\hfill
	\caption{\label{fig:geom_scatt}  Panel (a) is a schematic representation of a slab of cross like-holes which is finite in the $x_1$-direction ($N=10$ number of unit cells) and periodic in the $x_2$-direction. The red shadows represent the finite element domain $\gamma$ where the scattering problem is set up and solved. The domain $\gamma$ contains two perfectly matched layers (darker red regions at the sides of the rectangular domain $\gamma$). Panel (b) is a schematic representation of a unit cell $\Omega(n_1,n_2)$,  for a generic pair of integers $(n_1,n_2)$. The inner boundary $\partial^{(-)} \Omega(n_1,n_2)$ (dashed black line), the outer boundary $\partial^{(+)} \Omega(n_1,n_2)$ (dot-dashed black line), and the normal and traction vectors along  $\partial^{(-)} \Omega(n_1,n_2)$ (black and red arrow lines, respectively) are also shown.} 
\end{figure}

\begin{table}[H]
	\centering
	\begin{tabular}{ccc}
		$a$ & $b$ & $c$ \\[1mm]
		\hline 
		\noalign{\vskip1mm}
		$\left[\mathrm{mm}\right]$ & $\left[\mathrm{mm}\right]$ & $\left[\mathrm{mm}\right]$ \\[1mm]
		\hline 
		\hline 
		\noalign{\vskip1mm}
		$1$ & $0.9$ & $0.3$ \\[1mm]
	\end{tabular}$\qquad\qquad$%
	\begin{tabular}{ccc}
		$\mu_{\rm Al}$ & $\lambda_{\rm Al}$&$\rho_{\rm Al}$\\[1mm]
		\hline 
		\noalign{\vskip1mm}
		$\left[\mathrm{GPa}\right]$ & $\left[\mathrm{GPa}\right]$&$\left[\mathrm{kg/m^3}\right]$\\[1mm]
		\hline 
		\hline 
		\noalign{\vskip1mm}
		$26.32$ & $51.08$&$2700$  \\[1mm]
	\end{tabular}
	\caption{\label{tab:microstructureparams}Geometric parameters of the unit-cell (left), and its constitutive material (right).}
\end{table}

We consider an incident time-harmonic plane wave: 
\begin{equation}\label{eq:u_inc}
u^{j, \rm i}(x, t) = \bar{u}^{j, \rm i}(x)e^{- i\omega t} =d^j\, e^{i \langle k^j, x \rangle- i\omega t},
\end{equation}
where the index $j\in\{L,S\}$ denotes longitudinal and shear waves, respectively. Accordingly,\\ $k^j=\omega/c_j(\sin{\theta},-\cos{\theta},0)^{\rm T}$, with $j\in\{L,S\}$ and $\theta$ the angle of incidence,  $c_L=\sqrt{(\lambda_{\rm Al} + 2\mu_{\rm Al})/\rho_{\rm Al}}$ and $c_S=\sqrt{\mu_{\rm Al}/\rho_{\rm Al}}$ are the longitudinal and shear wave speeds for aluminum. The Lam\'e parameters $\lambda_{\rm Al}$, $\mu_{\rm Al}$ uniquely define the fourth order stiffness tensor $\mathbb{C}_{\rm Al}$, whose Voigt representation is 
\begin{equation}
\tilde{\mathbb{C}}_{\rm Al}=\begin{pmatrix}2\mu_{\rm Al}+\lambda_{\rm Al} & \lambda_{\rm Al} & \star & 0 & 0 & 0\\ \lambda_{\rm Al} & 2\mu_{\rm Al}+\lambda_{\rm Al} & \star & 0 & 0 & 0\\ \star & \star & \star & 0 & 0 & 0\\ 0 & 0 & 0 & \star & 0 & 0\\ 0 & 0 & 0 & 0 & \star & 0\\ 0 & 0 & 0 & 0 & 0 & \mu_{\rm Al} \end{pmatrix},
\end{equation}
where the stars denote the components which do not intervene in the plane-strain case.
In equation \eqref{eq:u_inc} we have introduced the polarization vectors $d^j$, $j\in\{L,S\}$ of amplitude $d_0$ for longitudinal and shear waves, defined as  $d^L=d_0(\sin \theta, -\cos \theta, 0)^{\rm T}$ and $d^S=d_0(\cos \theta, \sin \theta, 0)^{\rm T}$, respectively. The scattering problem in terms of the displacement field $u^j\equiv u^j(x,t)$, $j\in\{L,S\},$ in the  micro-structured material, according isotropic to linear elasticity, can be written as:\footnote{The index $j\in \{L,S\}$ for the elastic field $u^j$ indicates the fact that $u^j$ is the solution field generated by an $L$ or $S$ incident wave, respectively. Nevertheless, $u^j$ typically contains in itself coupled $L$ and $S$ waves as a result of scattering. } 
\begin{equation}\label{eq:PDEs_microstructured}
\begin{cases}
\rho_{\rm Al}~u^j_{,tt} = ~ {\rm Div}\left(\mathbb{C}_{\rm Al}~{\rm sym}~\nabla u^j \right),~~~~x \in\ \Omega_0({n_1,n_2})\\
\\
f(u^j) = 0, ~~~ x \in \partial ^{(-)} \Omega ({n_1,n_2}),~~~\forall n_1\in\{1,\cdots,N\}~~~{\rm and}~~~\forall n_2 \in \mathbb{Z},
\end{cases}
\end{equation}
where we have introduced the traction vectors
\begin{equation}\label{eq:traction}
f(u^j) = (\mathbb{C}_{\rm Al}~{\rm sym}~\nabla  u^j)\cdot \nu,~~~j\in\{L,S\},
\end{equation}  
$\nu$  being the normal unit vector (see black arrow line in Fig. \ref{fig:geom_scatt}(b))  to the cross-like holes boundaries $\partial ^{(-)} \Omega ({n_1,n_2})$ and where  we denote by $\Omega_0$ the part of the domain $\Omega$, which is non-empty (occupied by aluminum). The elastic field in \eqref{eq:PDEs_microstructured} can be written according to the scattering time-harmonic ansatz\footnote{We limit ourselves to the plane-strain case, so that we set $u^j_3=0$ and $u^j_{i,3}=0$, $j\in \{L,s\}$, $i \in \{1,2,3\}$.}     
\begin{equation}\label{eq:def_u_total}
u^j(x, t)= \left(\bar{u}^{ j,\rm i}(x)+\bar{u}^{\rm sc}(x)\right)e^{-i\omega t}, \quad j \in \{L,S\},
\end{equation}
where $ \bar{u}^{ j,\rm i}$ has been introduced in equation \eqref{eq:u_inc} and $\bar{u}^{\rm sc}$ is the so-called scattered solution, which typically contains coupled $L$ and $S$ scattered contributions. By linearity of the traction vector \eqref{eq:traction}, and using equation \eqref{eq:def_u_total}, we obtain
\begin{equation}\label{eq:def_f_total}
f(u^j(x, t)) = \left[ f(\bar{u}^{ j, \rm i}(x)) +f(\bar{u}^{\rm sc}(x))\right]e^{-i \omega t }=0, \quad j \in \{L,S\}.
\end{equation}
Using the fact that  $u^{ j,\rm i}(x,t)$ is a solution of the PDE in equation \eqref{eq:PDEs_microstructured}, together with equation \eqref{eq:def_f_total}, the PDEs system \eqref{eq:PDEs_microstructured} can be rewritten in a time-harmonic form with respect to the field $u^{ \rm sc}(x)$, as:
\begin{small}
	\begin{equation}\label{eq:PDEs_microstructured_sc}
	\begin{cases}
	-\omega^2\rho_{\rm Al}~\bar{u}^{\rm sc} = ~ {\rm Div}\left(\mathbb{C}_{\rm Al}~{\rm sym}~\nabla \bar{u}^{\rm sc} \right),~~~~x \in\ \Omega_0\\
	\\
	f(\bar{u}^{\rm sc}) \equiv (\mathbb{C}_{\rm Al}~{\rm sym}~\nabla  \bar{u}^{\rm sc})\cdot \nu= -f(\bar{u}^{ j, \rm i}), ~~~ x \in \partial ^{(-)} \Omega ({n_1,n_2}),~~~\forall n_1\in\{1,\cdots,N\}~~~{\rm and}~~~\forall n_2 \in \mathbb{Z},
	\end{cases}
	\end{equation}
\end{small}with $j\in\{L,S\}$ and where we have canceled out  time-harmonic factors. The analytical expressions for the boundary conditions for the scattered field (see right-hand side of the boundary conditions in equation \eqref{eq:PDEs_microstructured_sc}) are given by:\small
\begin{align}\label{eq:tractions-v-incident}
f(\bar{u}^{\rm sc})=f_{\rm v}^j(x)&:=  (\mathbb{C}_{\rm Al}~{\rm sym}~\nabla \bar{u}^{j,\rm i})\cdot e_1 \nonumber \\
&\hspace{0.12cm}= \frac{i\omega\rho}{c_j}\left\{\left[c_L^2\sin{\theta}\,  
d_1^{j, \rm i}
-(c_L^2-2c_S^2)\cos{\theta}\,
d_2^{j, \rm i}
\right]e_1 + c_S^2\left[ -\cos{\theta}\,
d_1^{j, \rm i}+\sin \theta \,d_2^{j, \rm i} \right]e_2 \right\}e^{i  \langle k^{j, \rm i} , x\rangle },
\end{align}\normalsize
with $j \in \{L,S\}$ for vertical boundaries  of  $\partial^{(-)} \Omega(n_1,n_2)$ with normal vector parallel to $e_1$ (see Fig. \ref{fig:geom_scatt}(b)). Similarly, for vertical boundaries with normal vector anti-parallel to $e_1$ we have $f(\bar{u}^{j, \rm i})=-f_{\rm v}^j(x)$. In addition, for horizontal  boundaries in $\partial^{(-)} \Omega(n_1,n_2)$ whose normal vector is parallel to $e_2$ we have: \small
\begin{align}\label{eq:tractions-h-incident}
\hspace{-10pt}f(\bar{u}^{\rm sc}) = f_{\rm h}^j(x)&:=  (\mathbb{C}_{\rm Al}~{\rm sym}~\nabla \bar{u}^{j,\rm i})\cdot e_2 \nonumber \\
&\hspace{0.12cm}= \frac{i\omega\rho}{c_j}\left\{ \left[(c_L^2-2c_S^2)\sin{\theta}\,15
d_1^{j, \rm i}
-c_L^2\cos{\theta}\,
d_2^{j, \rm i}
\right]e_2 + c_S^2\left[ -\cos{\theta}\,
d_1^{j, \rm i}+\sin\theta \,d_2^{j,\rm i}\right] e_1 \right\}e^{i  \langle k^{j,\rm i} , x\rangle },
\end{align}\normalsize
with $j \in \{L,S\}$.
Similarly, for vertical boundaries with normal vector anti-parallel to $e_2$ we have $f(\bar{u}^{\rm sc})=-f_{\rm h}^j(x)$.

\subsection{Bloch-Floquet conditions}
We recall that the primitive vectors of a square lattice are defined as:
\begin{equation}
a_1 = a \,e_1,\quad \text{and} \quad a_2=a \,e_2,
\end{equation}
where $a$ is the side of the unit-cell (see Fig. \ref{fig:geom_scatt}(b)). Since the scatterers (\emph{i.e.} the cross-like holes) are periodic in the $x_2$-direction, the displacement field in equation \eqref{eq:def_u_total} satisfies Bloch-Floquet boundary conditions
\begin{equation}\label{eq:def_BF_x2}
\bar{u}^{ \rm sc}(x+n_2 ~a_2) = e^{i n_2 k_2 a }\, \bar{u}^{ \rm sc}(x),~~~{\rm for}~~~x\in\gamma,~~~{\rm and}~~~n_2\in \mathbb{Z}, 
\end{equation}
where $k_2$ is the component along the $x_2$-direction of the wave vector $ k$. The value of $k_2$ is known and should be equal, given the considered geometry represented in Fig. \ref{fig:geom_scatt}(a), to the second component of the wave vector of the incident wave. This requirement is essential in order to construct a solution which satisfies the prescribed boundary conditions within a micro-structured medium which is periodic in one dimension, \emph{i.e.} in a layered micro-structured medium (see e.g. \cite{platts2002two}). The conservation of the value of $k_2$ provides the mathematical justification of Snell's law governing the refraction of waves at the interface between two half-spaces with different material parameters (see the book by Leckner \cite{leckner2016theory} for a mathematical introduction encompassing several physical scenarios). As it is customary in Floquet theory of PDEs with periodic coefficients, we can obtain the solution of the PDEs system \eqref{eq:PDEs_microstructured_sc}, by solving the problem in its period (\emph{i.e.} the red strip in Fig. \ref{fig:geom_scatt}(a) here denoted as $\gamma$) provided that the Bloch-Floquet condition \eqref{eq:def_BF_x2} on the scattered field is satisfied.    

Although the $x_1$ extension of the homogeneous part of the domain $\gamma$ is infinite in our model problem, in the finite-element implementation of the boundary value problem we are of course restricted to finite computational domains. In order to annihilate the reflection from the sides of $\gamma$ with constant $x_1$, we use perfectly matched layers \cite{basu2003perfectly} away from the microstructure.

\subsection{Reflectance}
The time-averaged Poynting vector associated with a 2D time-harmonic displacement field\\ $u(x,t)=\bar{u}(x)\,e^{-i\omega t}$ is defined as \cite{auld1973acoustic1,aivaliotis2018low}
\begin{equation}\label{eq:flux_C}
F=-\frac{\omega}{2} \,{\rm Re} \left(i \sigma \cdot{u}^*\right),
\end{equation}\\
where ``$*$'' denotes complex conjugation and $\sigma=\mathbb{C}_{\rm Al}\sym \nabla  u$ is the Cauchy stress tensor associated with the elastic field $u$. From equation \eqref{eq:flux_C}, and using equation \eqref{eq:u_inc}, it follows that the energy flux associated with the incident displacement field \eqref{eq:u_inc} is 
\begin{equation}
{ F}^{j,\rm i}=\frac{1}{2} \rho\, c_j |d^{j,\rm i}|^2 \frac{ k^{j,\rm i}}{|k^{j,\rm i}|}, \quad j\in\{L,S\}.
\end{equation}
Similarly, we define the flux $F^{ \rm sc}$ of the scattered field to be as in equation \eqref{eq:flux_C} with\\ $u(x,t)=\bar{u}^{ \rm sc}(x)e^{-i\omega t}$, where $\bar{u}^{ \rm sc}$ is the solution of the  PDEs system \eqref{eq:PDEs_microstructured_sc}). The reflectance, \emph{i.e.} the ratio of reflected energy and incident energy passing through a vertical line of length $a$,  is 
\begin{align}\label{eq:def_refl}
{\cal R}^j &= \frac{1}{ \langle F^{j,\rm i}, e_1 \rangle }  \frac{1}{a}\int_{-a/2}^{a/2} \langle \left.{ F}^{\rm sc}\right|_{x_1\ll -aN/2}, e_1\rangle\,{\rm d}x_2,\nonumber\\
&=\frac{2}{\rho\, c_j |d^{j,\rm i}|^2 \sin(\theta)} \frac{1}{a}\int_{-a/2}^{a/2} \langle \left.F^{\rm sc}\right|_{x_1\ll -aN/2}, e_1\rangle\,{\rm d}x_2,~~~j\in\{L,S\}.
\end{align}
The reflectance associated with incident shear waves (${\cal R}^S$) differs from that associated with incident longitudinal waves (${\cal R}^{L}$). The Poynting vector is evaluated at a given  $x_1\ll -a N/ 2$ away from the slab to avoid the contribution from the near elastic field close to the microstructure. Provided the condition $x_1\ll -a N/ 2$ is satisfied,  we have verified that the reflectance \eqref{eq:def_refl} does not depend on the exact value of $x_1$.

\section{Results and discussion}\label{sec:results}
In this section we present the comparison between the refractive behavior of the finite microstructured metamaterial's slab (see Fig.\ref{fig:geom_scatt}) and the relaxed micromorphic model (see Fig. \ref{fig:interfaces}). We also provide the results concerning the relaxed micromorphic single interface, which will be seen as an average behavior with respect to the micromorphic slab of finite size. To this end, we chose the material parameters of the relaxed micromorphic model as in Table \ref{tab:Numerical values}.

\begin{table}[H]
	\begin{centering}
		\begin{tabular}{ccccc}
			\hspace{-1cm}
			\begin{tabular}{ccc}
				$\lambda_e$ & $\mu_{e}$ & $\mu_{e}^{*}$\tabularnewline[1mm]
				\hline 
				\noalign{\vskip1mm}
				$\left[\textrm{GPa}\right]$ & $\left[\textrm{GPa}\right]$ & $\left[\textrm{GPa}\right]$\tabularnewline[1mm]
				\hline 
				\hline 
				\noalign{\vskip1mm}
				$2.33$ & $10.92$ & $0.67$\tabularnewline[1mm]
			\end{tabular} &   %
			\begin{tabular}{ccc}
				$\lambda_{\mathrm{micro}}$ & $\mu_{\mathrm{micro}}$ & $\mu_{\mathrm{micro}}^{*}$\tabularnewline[1mm]
				\hline 
				\noalign{\vskip1mm}
				$\left[\textrm{GPa}\right]$ & $\left[\textrm{GPa}\right]$ & $\left[\textrm{GPa}\right]$\tabularnewline[1mm]
				\hline 
				\hline 
				\noalign{\vskip1mm}
				$5.27$ & $12.8$ & $8.33$\tabularnewline[1mm]
			\end{tabular}&  
			\begin{tabular}{c}
				$\mu_{c}$ \tabularnewline[1mm]
				\hline 
				\noalign{\vskip1mm}
				$\left[\textrm{GPa}\right]$\tabularnewline[1mm]
				\hline 
				\hline 
				\noalign{\vskip1mm}
				$2.28\cdot10^{-3}$\tabularnewline[1mm]
			\end{tabular}& 	\begin{tabular}{c}
				$L_{c}$ \tabularnewline[1mm]
				\hline 
				\noalign{\vskip1mm}
				$\left[\textrm{m}\right]$\tabularnewline[1mm]
				\hline 
				\hline 
				\noalign{\vskip1mm}
				$4.41 \cdot 10^{-6}$\tabularnewline[1mm]
			\end{tabular}&
			\begin{tabular}{ccc}
				$\lambda_{\mathrm{macro}}$ & $\mu_{\mathrm{macro}}$ & $\mu_{\mathrm{macro}}^{*}$\tabularnewline[1mm]
				\hline 
				\noalign{\vskip1mm}
				$\left[\textrm{GPa}\right]$ & $\left[\textrm{GPa}\right]$ & $\left[\textrm{GPa}\right]$\tabularnewline[1mm]
				\hline 
				\hline 
				\noalign{\vskip1mm}
				$1.74$ & $5.89$ & $0.62$\tabularnewline[1mm]
			\end{tabular}\tabularnewline
		\end{tabular}
		
		\begin{tabular}{ccccccccc}
			$\rho$ & $\eta_{1}$ & $\eta_{2}$ & $\eta_{3}$ & $\eta_{1}^{*}$&$\overline{\eta}_{1}$ & $\overline{\eta}_{2}$ & $\overline{\eta}_{3}$ & $\overline{\eta}_{1}^{*}$\tabularnewline[1mm]
			\hline 
			\noalign{\vskip1mm}
			$\left[\mathrm{kg/m}^{3}\right]$ & $\left[\mathrm{kg/m}\right]$ & $\left[\mathrm{kg/m}\right]$ & $\left[\mathrm{kg/m}\right]$ & $\left[\mathrm{kg/m}\right]$&$\left[\mathrm{kg/m}\right]$ & $\left[\mathrm{kg/m}\right]$ & $\left[\mathrm{kg/m}\right]$ & $\left[\mathrm{kg/m}\right]$\tabularnewline[1mm]
			\hline 
			\hline 
			\noalign{\vskip1mm}
			$1485$ & $8.6\cdot10^{-5}$ & $10^{-7}$ & $-2.2\,\cdot10^{-5}$ & $3.3\cdot10^{-5}$&$5.6 \cdot 10^{-5}$ & $7.3 \cdot10^{-4}$  & $1.7 \cdot 10^{-4}$& $9\cdot10^{-7}$\tabularnewline[1mm]
		\end{tabular}

		\par\end{centering}
		\caption{\label{tab:Numerical values}Summary of the numerical values for the elastic (top) and inertia (bottom) parameters of the tetragonal relaxed micromorphic model in the 2D plane-strain case. The macroscopic parameters of the resulting homogenized Cauchy material (see \cite{dagostino2019effective, aivaliotis2018low}) are also provided (top right).}
\end{table}

We will show that the performed fitting allows the description of the average reflectance of the metamaterial slab even if the Bloch-Floquet patterns are not exactly reproduced point by point (see Fig. \ref{fig:DispersionCurves}). The micromorphic modeling framework presented here will need further generalizations in order to exactly match these patterns. Nevertheless, the astonishingly pertinent results presented in section \ref{sec:slabscattering} clearly indicate that the present paper must be used as a starting stepping stone for a new mindset in the modeling of metamaterials, which will enable the design of complex large-scale meta-structures.  

The choice of the metamaterial parameters is made according to the procedure presented in \cite{dagostino2019effective,neff2019identification}, which allows to determine the parameters of the relaxed micromorphic model on a specific metamaterial by an inverse approach. This fitting procedure is based on the determination of the elastic parameters of the relaxed micromorphic model via numerical static tests on the unit cell and of the remaining inertia parameters via a simple inverse fitting of the dispersion curves on the analogous dispersion patterns as obtained by Bloch-Floquet analysis (see \cite{dagostino2019effective,neff2019identification} for details). The fitting of the bulk dispersion curves for the periodic metamaterial, whose unit cell is shown in Fig. \ref{fig:geom_scatt2}, is presented in Fig. \ref{fig:DispersionCurves}.

\begin{figure}[H]
	\vspace{0.4cm}
	\begin{subfigure}{.5\textwidth}
		\centering
		\includegraphics[scale=0.6]{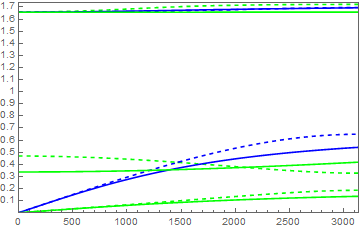}
		\vspace{0.3cm}
		\caption{}
		\vspace{-0.3cm}
		\label{fig:DispersionCurvesNormal}
	\end{subfigure}%
	\begin{subfigure}{.55\textwidth}
		\centering
		\includegraphics[scale=0.6]{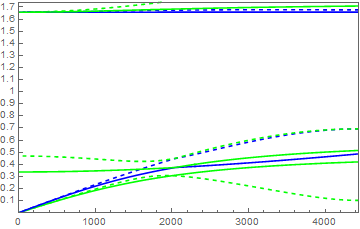}
			\vspace{0.3cm}
		\caption{}
		\vspace{-0.3cm}
		\label{fig:DispersionCurves45}
	\end{subfigure}
	\caption{\small Dispersion diagrams for normal (a) and 45 degrees (b) incidence. The solid curves are obtained via the tetragonal anisotropic relaxed micromorphic model, while the dashed curves are issued by Bloch-Floquet analysis. In panel (a), green color stands for modes which are mostly activated by a shear incident wave, while blue color indicates modes which are mostly activated by a longitudinal incident wave. This uncoupling between L and S activated modes at normal incidence is analytically checked in the relaxed micromorphic model and only approximate for Bloch-Floquet modes. In panel (b), we keep the same coloring, but all curves are coupled together, which means that L and S incident waves may simultaneously activate all modes at a given frequency.	}
	\label{fig:DispersionCurves}
	\begin{picture}(0,0)(0,0)
	\put(50,254) {\scriptsize Dispersion curves for normal incidence}
	\put(-7,145) {\rotatebox{90}{\scriptsize{Wave frequency $[10^7 \text{rad/s}]$}}}
	\put(63,105){\scriptsize{Wave number $\kabs = \sqrt{k_1^2+k_2^2}$}}
	\put(280,254) {\scriptsize Dispersion curves for $45$ degrees incidence}
	\put(232,145) {\rotatebox{90}{\scriptsize{Wave frequency $[10^7 \text{rad/s}]$}}}
	\put(300,105){\scriptsize{Wave number $\kabs = \sqrt{k_1^2+k_2^2}$}}
	\end{picture}
\end{figure}
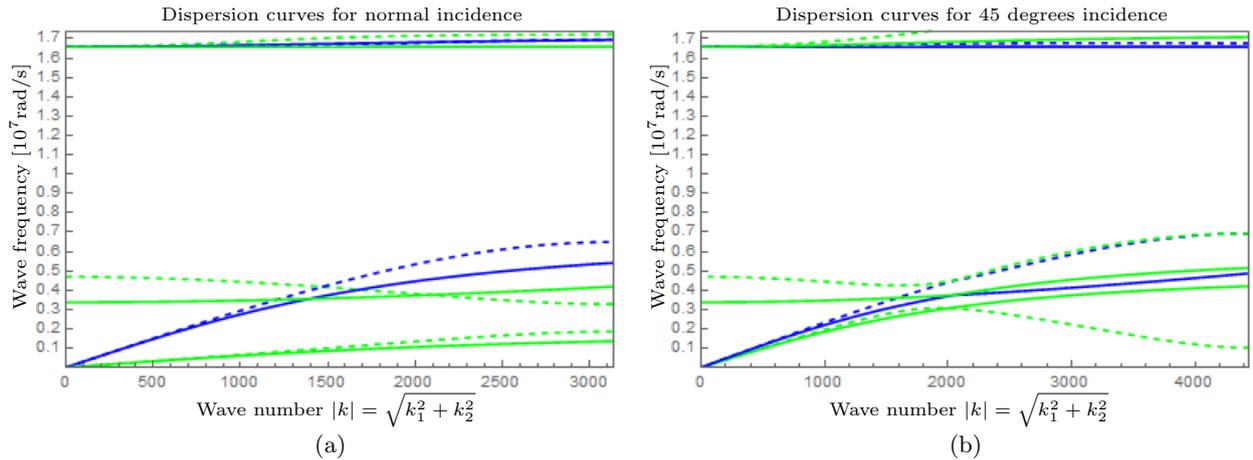

We checked that a sort of distinction between modes which are activated by an L or S incident wave can be made for an incident wave which is orthogonal to the interface. This is exactly true for the relaxed micromorphic model, but an analogous trend can be found for Bloch-Floquet modes at least for the lower frequency modes (before the band-gap). The uncoupling between L and S activated modes is present only for $\theta =0$ (and, by symmetry, $\theta = \pi/2$), but is lost for any other direction of propagation. In general, for any given frequency, all modes which are pertinent at that frequency may be simultaneously activated by an L or S incident wave (excluding the particular case of normal incidence). Nevertheless, we will show that this uncoupling hypothesis can be retained with little error for angles of incidence which are close to normal incidence.


\subsection{Scattering at a relaxed micromorphic slab}\label{sec:slabscattering}

We start by comparing the reflection coefficient of the relaxed micromorphic slab as a function of the frequency for two fixed directions of propagation of the incident wave ($\theta= \frac{\pi}{2}$ and $\theta= \frac{\pi}{4}$) and for both longitudinal and shear incident waves with the fully resolved linear elastic microstructured model.

\begin{figure}[h!]
	\hspace{-0.7cm}
	\begin{tabular}{c l}
		\includegraphics[scale=0.55]{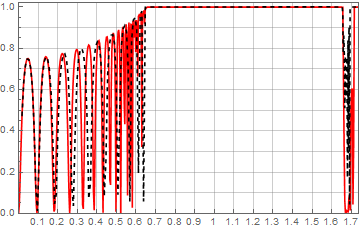}
		&
		\hspace{0.5cm}
		\includegraphics[scale=0.55]{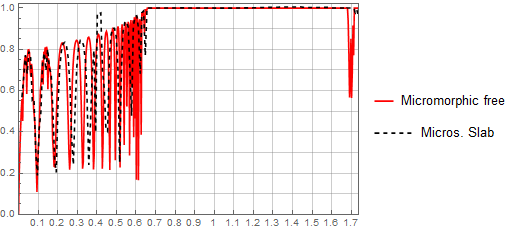}
		\\
		(a) & \hspace{3.7cm} \vspace{-0.2cm} (b)
	\end{tabular}
	\caption{\small Reflection coefficient at the relaxed micromorphic slab for an incident L wave and for two directions of propagation $\theta= \pi/2$ (normal incidence) (a) and $\theta= \pi/4$ (b) and for the free microstructure boundary condition. The red curve is generated by the analytical tetragonal relaxed micromorphic model and the black dashed line indicates the fully resolved linear elastic microstructured model.}
	\label{fig:ReflSlabPfr}
	\begin{picture}(0,0)(0,0)
	\put(-3,210) {\scriptsize Reflection coefficient, L $\pi/2$ incidence, free microstr.}
	\put(-20,115) {\rotatebox{90}{\scriptsize{Reflection coefficient}}}
	\put(50,75){\scriptsize{Wave frequency $[10^7 \text{rad/s}]$}}
	\put(226,210) {\scriptsize Reflection coefficient, L $\pi/4$ incidence, free microstr.}
	\put(208,115) {\rotatebox{90}{\scriptsize{Reflection coefficient}}}
	\put(280,75){\scriptsize{Wave frequency $[10^7 \text{rad/s}]$}}
	\end{picture}
\end{figure}
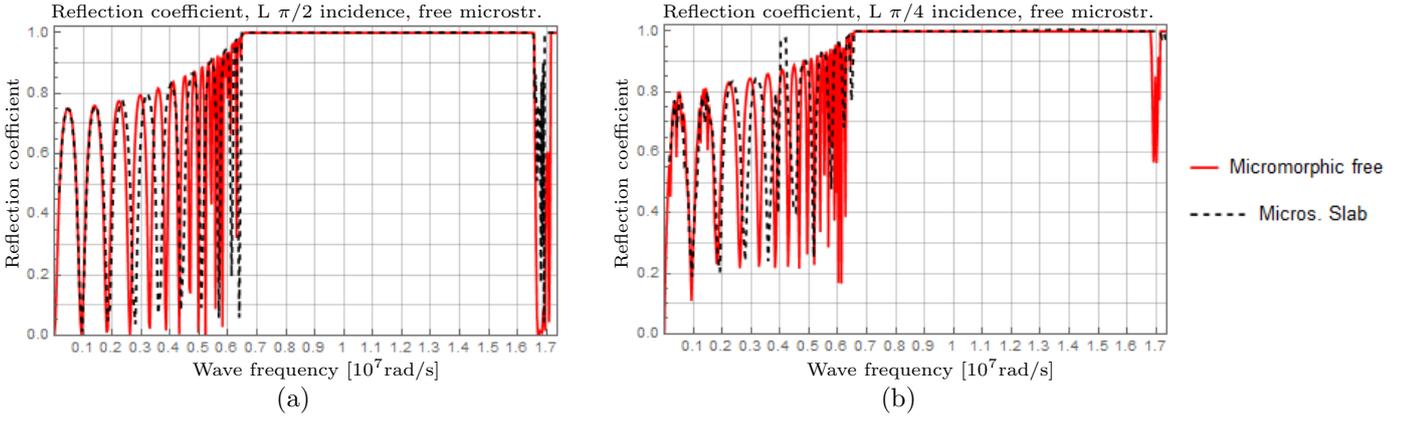

\begin{figure}[h]
		\hspace{-0.7cm}
	\begin{tabular}{c l}
		\includegraphics[scale=0.55]{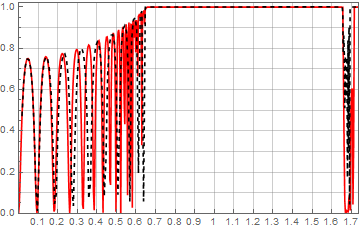}
		&
		\hspace{0.5cm}
		\includegraphics[scale=0.55]{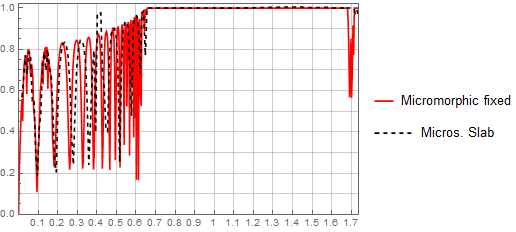}
		\\
		(a) & \hspace{3.7cm} \vspace{-0.2cm} (b)
	\end{tabular}
	\caption{\small Reflection coefficient at the relaxed micromorphic slab for an incident L wave and for two directions of propagation $\theta= \pi/2$ (normal incidence) (a) and $\theta= \pi/4$ (b) and for the fixed microstructure boundary condition. The red curve is generated by the analytical tetragonal relaxed micromorphic model and the black dashed line indicates the fully resolved linear elastic microstructured model.}
	\label{fig:ReflSlabPfi}
	\begin{picture}(0,0)(0,0)
	\put(-3,210) {\scriptsize Reflection coefficient, L $\pi/2$ incidence, fixed microstr.}
	\put(-20,115) {\rotatebox{90}{\scriptsize{Reflection coefficient}}}
	\put(50,75){\scriptsize{Wave frequency $[10^7 \text{rad/s}]$}}
	\put(226,210) {\scriptsize Reflection coefficient, L $\pi/4$ incidence, fixed microstr.}
	\put(208,115) {\rotatebox{90}{\scriptsize{Reflection coefficient}}}
	\put(280,75){\scriptsize{Wave frequency $[10^7 \text{rad/s}]$}}
	\end{picture}
\end{figure}

\begin{figure}[h]
	\hspace{-0.7cm}
	\begin{tabular}{c l}
		\includegraphics[scale=0.55]{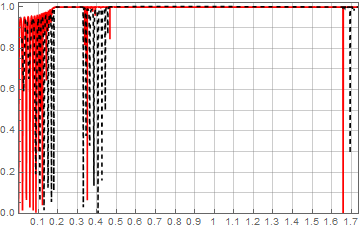}
		&
		\hspace{0.5cm}
		\includegraphics[scale=0.55]{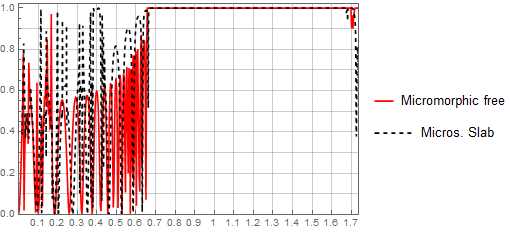}
		\\
		(a) & \hspace{3.7cm} \vspace{-0.2cm} (b)
	\end{tabular}
	\caption{\small Reflection coefficient at the relaxed micromorphic slab for an incident S wave and for two directions of propagation $\theta= \pi/2$ (normal incidence) (a) and $\theta= \pi/4$ (b) and for the free microstructure boundary condition. The red curve is generated by the analytical tetragonal relaxed micromorphic model and the black dashed line indicates the fully resolved linear elastic microstructured model.}
	\label{fig:ReflSlabSfr}
	\begin{picture}(0,0)(0,0)
	\put(-3,210) {\scriptsize Reflection coefficient, S $\pi/2$ incidence, free microstr.}
	\put(-20,115) {\rotatebox{90}{\scriptsize{Reflection coefficient}}}
	\put(50,75){\scriptsize{Wave frequency $[10^7 \text{rad/s}]$}}
	\put(226,210) {\scriptsize Reflection coefficient, S $\pi/4$ incidence, free microstr.}
	\put(208,115) {\rotatebox{90}{\scriptsize{Reflection coefficient}}}
	\put(280,75){\scriptsize{Wave frequency $[10^7 \text{rad/s}]$}}
	\end{picture}
\end{figure}

\begin{figure}[h!]
	\hspace{-0.7cm}
	\begin{tabular}{c l}
		\includegraphics[scale=0.55]{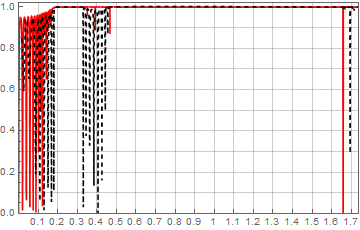}
		&
		\hspace{0.5cm}
		\includegraphics[scale=0.55]{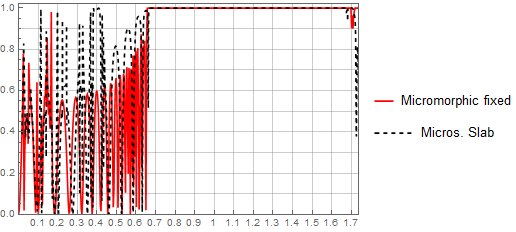}
		\\
		(a) & \hspace{3.7cm} \vspace{-0.2cm} (b)
	\end{tabular}
	\caption{\small Reflection coefficient at the relaxed micromorphic slab for an incident S wave and for two directions of propagation $\theta= \pi/2$ (normal incidence) (a) and $\theta= \pi/4$ (b) and for the fixed microstructure boundary condition. The red curve is generated by the analytical tetragonal relaxed micromorphic model and the black dashed line indicates the fully resolved linear elastic microstructured model.}
	\label{fig:ReflSlabSfi}
	\begin{picture}(0,0)(0,0)
	\put(-3,210) {\scriptsize Reflection coefficient, S $\pi/2$ incidence, fixed microstr.}
	\put(-20,115) {\rotatebox{90}{\scriptsize{Reflection coefficient}}}
	\put(50,75){\scriptsize{Wave frequency $[10^7 \text{rad/s}]$}}
	\put(226,210) {\scriptsize Reflection coefficient, S $\pi/4$ incidence, fixed microstr.}
	\put(208,115) {\rotatebox{90}{\scriptsize{Reflection coefficient}}}
	\put(280,75){\scriptsize{Wave frequency $[10^7 \text{rad/s}]$}}
	\end{picture}
\end{figure}
Figures \ref{fig:ReflSlabPfr} and \ref{fig:ReflSlabPfi} show the behavior of the reflection coefficient for the considered microstructured slab in the case of a longitudinal incident wave for the two boundary conditions (free and fixed microstructure) at normal incidence (Fig. \ref{fig:ReflSlabPfr}(a), \ref{fig:ReflSlabPfi}(a)) as well as for incidence at $45^\circ$  (Fig. \ref{fig:ReflSlabPfr}(b), \ref{fig:ReflSlabPfi}(b)). The black dashed line is the solution issued via the microstructured model obtained by coding all the details of the unit cell presented in section \ref{sec:Comsol}. The red continuous line is obtained by solving the relaxed micromorphic problem presented in section \ref{sec:RTSlab}, using the software Mathematica.

It is immediately evident that the relaxed micromorphic model is able to capture the overall behavior of the reflection coefficient for a very wide range of frequencies. More particularly, for lower frequencies and up to the band-gap region, oscillations of the reflection coefficient due to the finite size of the slab are observed in both models. The  band-gap region is also correctly described and corresponds to the frequency interval for which complete reflection ($\mathcal{R}=1$) is observed.

Few differences can be found between the two types of imposed boundary conditions (free and fixed microstructure). This indicates that the characteristic length $L_c$ has little effect in the dynamic regime and it remains a quantity that mainly influences the static behavior of the metamaterial by allowing to obtain the homogenization formulas for the identification of static micromorphic stiffnesses \cite{dagostino2019effective,neff2019identification}.

After the band-gap region, a characteristic frequency can be identified corresponding to which almost complete transmission occurs. This phenomenon is related to internal resonances at the level of the microstructure. It is easy to see that the relaxed micromorphic model is able to correctly describe also the internal resonance phenomenon. The internal resonance is clearly visible in both the discrete and the relaxed micromorphic model at normal incidence. It is, however, lost in the discrete simulation at $45^{\circ}$ notwithstanding the presence of a zero group velocity mode in both models. 

Similar arguments can be carried out for an S incident wave with reference to Figs. \ref{fig:ReflSlabSfr} and \ref{fig:ReflSlabSfi}. In addition to the comments carried out for an L incident wave, we remark in Fig. \ref{fig:ReflSlabSfr}(a) the presence of a small frequency interval, in which some transmission takes place correspondingly to the lower part of the band-gap. This small transmission interval is related to the S-incident activation of the first optic mode that can be associated to micro-motions (mainly rotations) at the level of the unit cell.

In the relaxed micromorphic model, the first optic mode is flat if we set $L_c = 0$, which means that there would be no propagation in the frequency interval of interest. On the other hand, setting $L_c>0$ gives a non-zero group velocity to the first optic mode, thus allowing to catch some transmission in this frequency interval (see Figures \ref{fig:ReflSlabSfr}(a) and \ref{fig:ReflSlabSfi}(a)). Being directly activated by $L_c$, this small transmission interval has a pattern which is evidently affected by the choice of the free or fixed microstructure boundary conditions. In order to improve the description of the transmission pattern in this single frequency interval, the relaxed micromorphic model needs to be enhanced with space-time non localities to include the possibility of describing negative group velocity patterns. 

Apart from the differences observed in the aforementioned small-frequency interval, changing the type of boundary conditions does not sensibly affect the transmission pattern for other frequencies. For this reason, we will present te following results only for the free-microstructure boundary conditions.

We depict in Figure \ref{fig:ReflSlabSweepFreeL} the transmission coefficient ($\mathcal{T} = 1 - \mathcal{R}$) as a function of both the angle and frequency of the longitudinal incident wave for both the relaxed micromorphic and the microstructured models. We observe an excellent agreement between the continuous and discrete simulation for frequencies lower than the band-gap. For those frequencies, transmission is principally allowed by the blue acoustic mode for all directions of propagation. Even if there exists some coupling at non-orthogonal incidence with the other lower frequency modes, the blue acoustic mode is the one which is predominantly activated by an L incident wave and it is, to a large extent, responsible for the transmission across the metamaterial slab before the band-gap. The band-gap region is also correctly described. The validation of the fitting performed on higher frequencies is less precise due to the difficulty encountered for properly controlling the convergence of the detailed microstructured simulations.

Figure \ref{fig:ReflSlabSweepFreeS} depicts the analogous results for an S incident wave. For frequencies lower than the band-gap, we observe once again an excellent agreement between the discrete and continuous simulations for all angles of incidence. We also remark additional interesting phenomena. Firstly, except for the small transmission frequency interval, we see that the band-gap region extends to lower frequencies when considering angles close to normal incidence. This is related to the previously discussed acoustic mode uncoupling, which is observed for angles close to normal incidence. A shear incident wave mostly activates the green acoustic mode (see Fig. \ref{fig:DispersionCurves}) and the first (green) optic one and they are almost entirely responsible for the propagation pattern before the band-gap region. Since the blue acoustic mode is not activated for angles close to normal incidence, the bottom band-gap limit is consequently lower compared to the case of an L incident wave.

A first threshold value of the angle of incidence can be identified (around $5\pi/24 $), for which the two (green and blue) acoustic modes start to couple and energy starts being transmitted. More remarkably, a second threshold value of the incident angle exists (around $\pi/3$), for which the amount of transmission suddenly increases, approaching a total transmission pattern. This is due to a stronger coupling between the two acoustic modes which are activated by an S incident wave for incident angles beyond the second threshold. This impressive pattern is clearly associated to the tetragonal symmetry of the metamaterial: the need for introducing ``generalized classes of symmetry'' in an enriched continuum environment becomes thus evident. We deduce that the agreement is very satisfactory for all the considered angles (going from normal incidence to incidence almost parallel to interface) and for the considered range of frequencies. This fact corroborates the hypothesis, which has been made according to Neumann's principle and which states that the class of symmetry of the metamaterial at the macroscopic scale is the same as the symmetry of the unit cell (tetragonal symmetry in this case).

We conclude this section by pointing out that the simulations performed to obtain Figures \ref{fig:ReflSlabSweepFreeL} and \ref{fig:ReflSlabSweepFreeS} took less than 1 hour for the relaxed micromorphic model and $3$ weeks for the discrete model. Both computations were made with $200$ points in the frequency range and for $90$ angles. 

This tremendous gain in computational time underlines the usefulness of an enriched continuum model versus a discrete one for the description of the mechanical behavior of finite-size metamaterials structures. Metamaterial characterization through enriched models of the micromorphic type opens the way to effective FEM implementation of morphologically complex meta-structures.

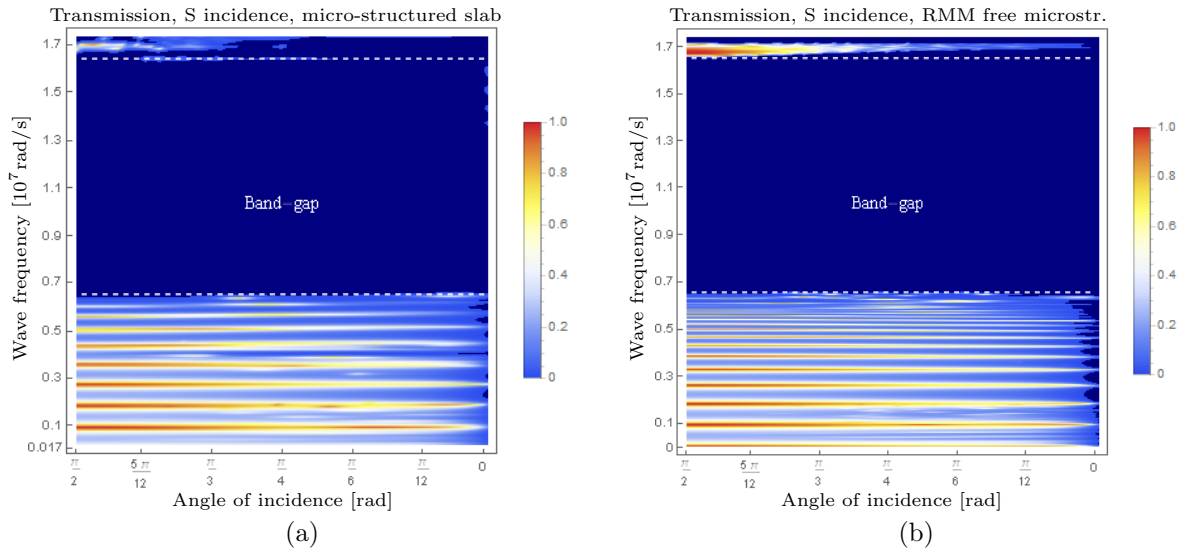
\begin{figure}[h!]
	\begin{tabular}{c l}
		\includegraphics[scale=0.487]{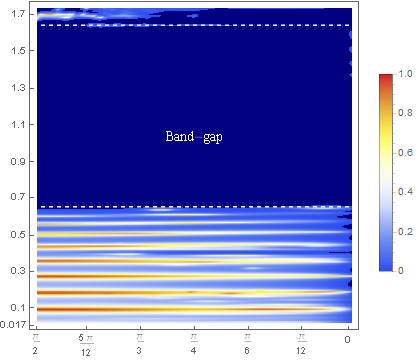}
		&
		\hspace{0.5cm}
		\includegraphics[scale=0.47]{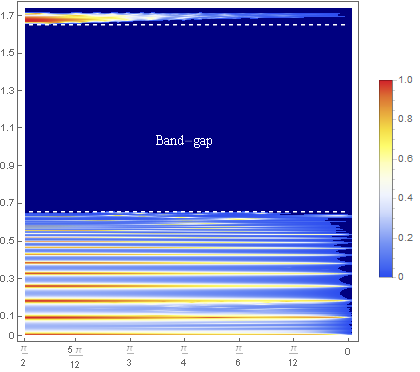}
		\\
		(a) & \hspace{3.7cm} \vspace{-0.2cm} (b)
	\end{tabular}
	\caption{\small\small Transmission coefficient of the metamaterial slab as a function of the angle of incidence $\theta$ and of the wave-frequency $\omega$ for an incident L wave and the free microstructure boundary condition. Panel (a) depicts the microstructured simulations, while panel (b) the analytical relaxed micromorphic model. The origin coincides with normal incidence ($\theta=\pi/2$), while the angle of incidence decreases towards the right until it reaches the value $\theta=0$, which corresponds to the limit case where the incidence is parallel to the interface. The band-gap region is highlighted by two dashed horizontal lines, where, as expected, we observe no transmission. The dark blue zone shows that no transmission takes place, while the gradual change from dark blue to red shows the increase of transmission, red being total transmission. }
	\label{fig:ReflSlabSweepFreeL}
	\begin{picture}(0,0)(0,0)
	\put(14,303) {\scriptsize Transmission, S incidence, micro-structured slab}
	\put(-1,175) {\rotatebox{90}{\scriptsize{Wave frequency $[10^7 \text{rad/s}]$}}}
	\put(60,120){\scriptsize{Angle of incidence $[\text{rad}]$}}
	\put(246,303) {\scriptsize Transmission, S incidence, RMM free microstr.}
	\put(230,175) {\rotatebox{90}{\scriptsize{Wave frequency $[10^7 \text{rad/s}]$}}}
	\put(288,120){\scriptsize{Angle of incidence $[\text{rad}]$}}
	\end{picture}
\end{figure}

\begin{figure}[H]
	\centering
		\begin{subfigure}{.44\textwidth}
			\centering
			\includegraphics[scale=0.47]{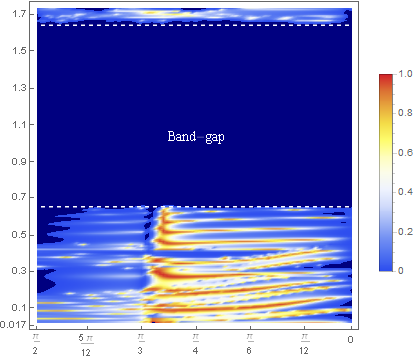}
			\caption{}
			\label{fig:ReflSlabSweepComsolS}
		\end{subfigure}%
	\begin{subfigure}{.4\textwidth}
		\centering
		\includegraphics[scale=0.47]{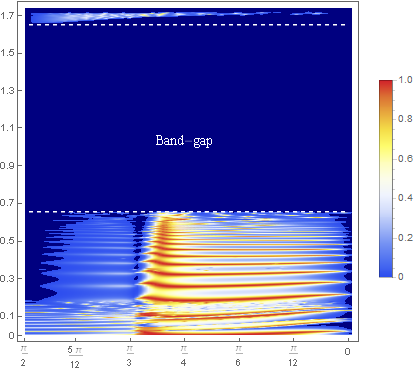}
		\caption{}
		\label{fig:ReflSlabSweepRMMS}
	\end{subfigure}
	\caption{\small Transmission coefficient of the metamaterial slab as a function of the angle of incidence $\theta$ and of the wave-frequency $\omega$ for an incident S wave and the free microstructure boundary condition. Panel (a) depicts the microstructured simulations, while panel (b) the analytical relaxed micromorphic model. The origin coincides with normal incidence ($\theta=\pi/2$), while the angle of incidence decreases towards the right until it reaches the value $\theta=0$, which corresponds to the limit case where the incidence is parallel to the interface. The band-gap region is highlighted by two dashed horizontal lines, where, as expected, we observe no transmission. The dark blue zone shows that no transmission takes place, while the gradual change from dark blue to red shows the increase of transmission, red being total transmission.}
	\label{fig:ReflSlabSweepFreeS}
\end{figure}

\subsection{Scattering at a single relaxed micromorphic interface}

In this subsection we show the results for the reflection coefficient obtained by using the single interface boundary conditions for the relaxed micromorphic continuum as described in section \ref{sec:BCsingle}.

Figure \ref{fig:ReflSinglePfr} shows the reflection coefficient as a function of the frequency for two different angles of incidence ($\theta = \pi/2$ and $\theta=\pi/4$), when considering an L incident wave for the ``single interface'' boundary conditions. Figure \ref{fig:ReflSingleSfr} shows the analogous results for an incident S wave.
As expected, the solution obtained using the ``single interface'' boundary conditions, provide a sort of average behavior for the oscillations at lower frequencies. This is sensible, since when considering a semi-infinite metamaterial, multiple reflections on the two boundaries of the slab are not accounted for. The difference between ``single'' and ``double interface'' boundary conditions in the relaxed micromorphic model becomes less pronounced for higher frequencies, since the wavelength of the considered waves is expected to be much lower than the characteristic-size of the slab.

\begin{figure}[h!]
	\hspace{-0.7cm}
	\begin{tabular}{c l}
		\includegraphics[scale=0.55]{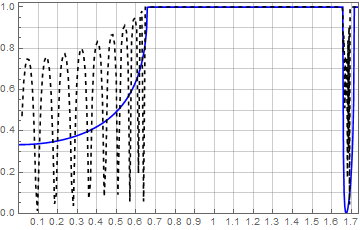}
		&
		\hspace{0.5cm}
		\includegraphics[scale=0.55]{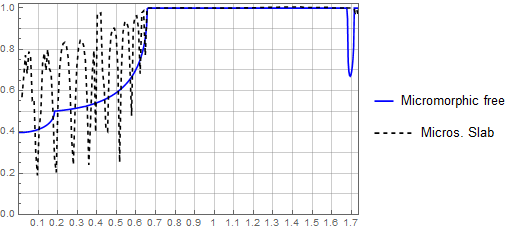}
		\\
		(a) & \hspace{3.7cm} \vspace{-0.2cm} (b)
	\end{tabular}
	\caption{\small Reflection coefficient at the single interface for an incident L wave and for two directions of propagation $\theta= \pi/2$ (normal incidence) (a) and $\theta= \pi/4$ (b) and for the free microstructure boundary condition. The blue curve is generated by the analytical tetragonal relaxed micromorphic model and the black dashed line indicates the fully resolved linear elastic microstructured model.}
	\label{fig:ReflSinglePfr}
	\begin{picture}(0,0)(0,0)
	\put(-3,210) {\scriptsize Reflection coefficient, L $\pi/2$ incidence, free microstr.}
	\put(-20,115) {\rotatebox{90}{\scriptsize{Reflection coefficient}}}
	\put(50,75){\scriptsize{Wave frequency $[10^7 \text{rad/s}]$}}
	\put(226,210) {\scriptsize Reflection coefficient, L $\pi/4$ incidence, free microstr.}
	\put(208,115) {\rotatebox{90}{\scriptsize{Reflection coefficient}}}
	\put(280,75){\scriptsize{Wave frequency $[10^7 \text{rad/s}]$}}
	\end{picture}
\end{figure}
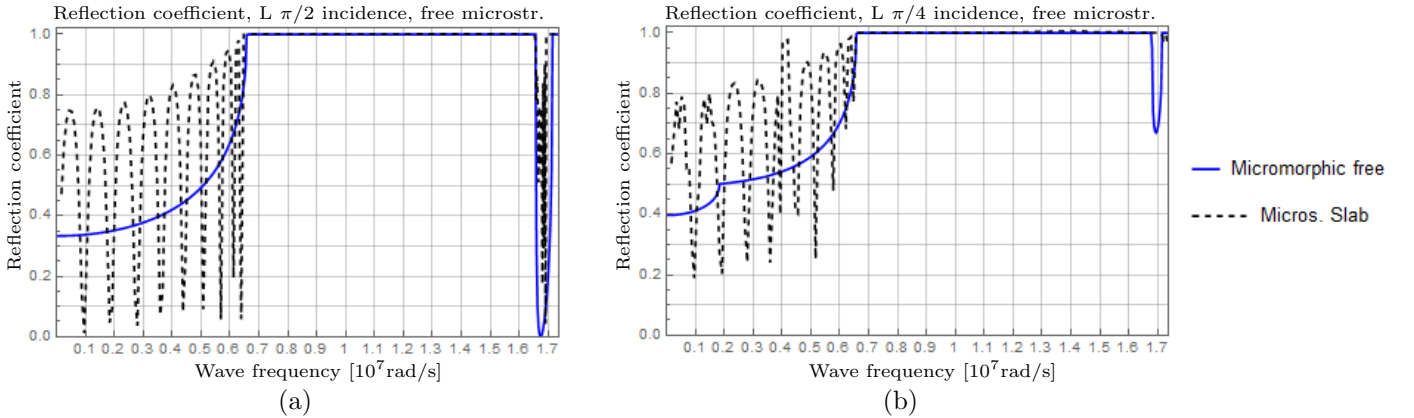

\begin{figure}[h!]
	\hspace{-0.7cm}
	\begin{tabular}{c l}
		\includegraphics[scale=0.55]{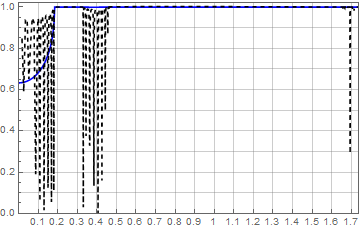}
		&
		\hspace{0.5cm}
		\includegraphics[scale=0.55]{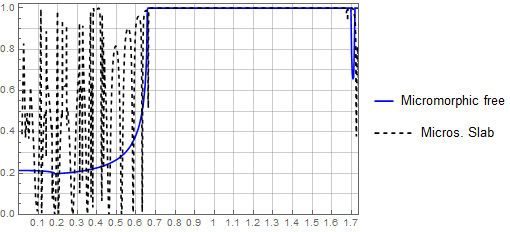}
		\\
		(a) & \hspace{3.7cm} \vspace{-0.2cm} (b)
	\end{tabular}
	\caption{\small Reflection coefficient at the single interface for an incident L wave and for two directions of propagation $\theta= \pi/2$ (normal incidence) (a) and $\theta= \pi/4$ (b) and for the free microstructure boundary condition. The blue curve is generated by the analytical tetragonal relaxed micromorphic model and the black dashed line indicates the fully resolved linear elastic microstructured model.}
	\label{fig:ReflSingleSfr}
	\begin{picture}(0,0)(0,0)
	\put(-3,210) {\scriptsize Reflection coefficient, S $\pi/2$ incidence, free microstr.}
	\put(-20,115) {\rotatebox{90}{\scriptsize{Reflection coefficient}}}
	\put(50,75){\scriptsize{Wave frequency $[10^7 \text{rad/s}]$}}
	\put(226,210) {\scriptsize Reflection coefficient, S $\pi/4$ incidence, free microstr.}
	\put(208,115) {\rotatebox{90}{\scriptsize{Reflection coefficient}}}
	\put(280,75){\scriptsize{Wave frequency $[10^7 \text{rad/s}]$}}
	\end{picture}
\end{figure}

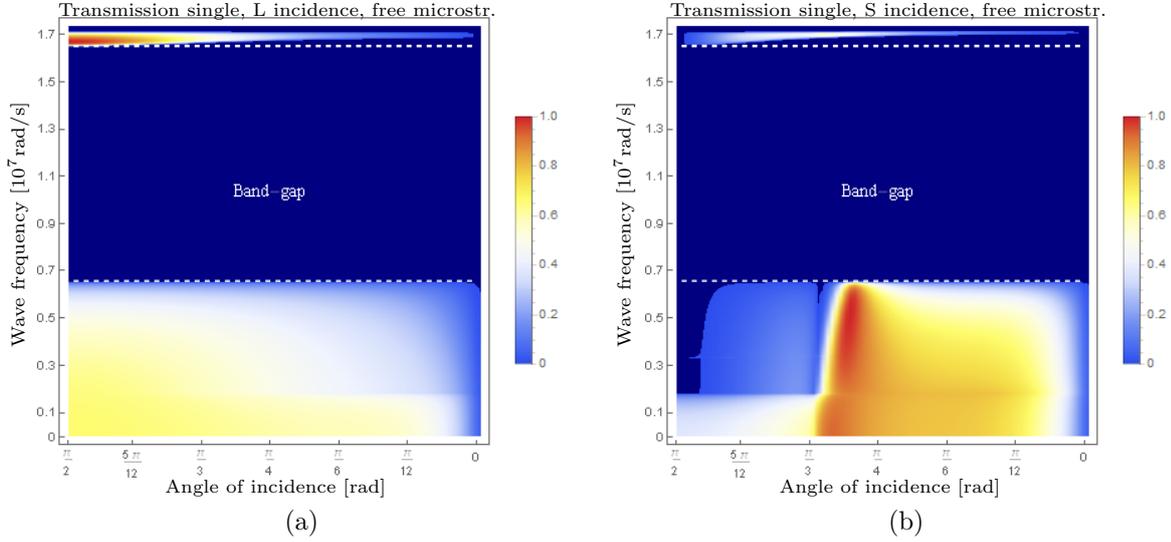
\begin{figure}[h!]
	\begin{tabular}{c l}
		\includegraphics[scale=0.47]{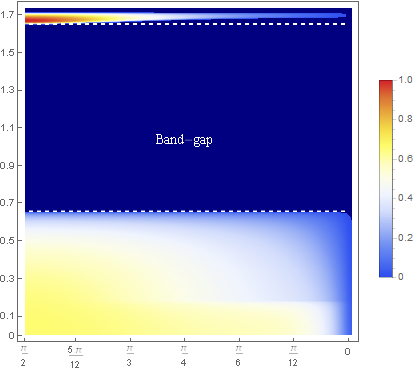}
		&
		\hspace{0.5cm}
		\includegraphics[scale=0.47]{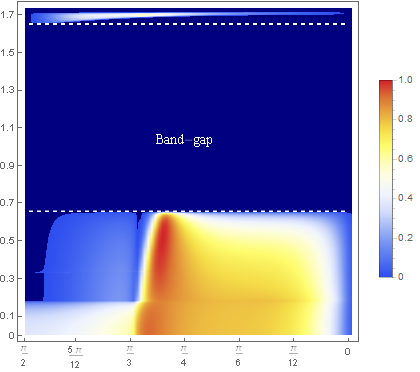}
		\\
		(a) & \hspace{3.7cm} \vspace{-0.2cm} (b)
	\end{tabular}
	\caption{\small\small Transmission coefficient of the single interface as a function of the angle of incidence $\theta$ and of the wave-frequency $\omega$. Panel (a) depicts the case of an incident L wave, while panel (b) the case of an incident S wave. The origin coincides with normal incidence ($\theta=\pi/2$), while the angle of incidence decreases towards the right until it reaches the value $\theta=0$, which corresponds to the limit case where the incidence is parallel to the interface. The band-gap region is highlighted by two dashed horizontal lines, where, as expected, we observe no transmission. The dark blue zone shows that no transmission takes place, while the gradual change from dark blue to red shows the increase of transmission, red being total transmission.}
	\label{fig:ReflSweepSingle}
	\begin{picture}(0,0)(0,0)
	\put(14,290) {\scriptsize Transmission single, L incidence, free microstr.}
	\put(-5,165) {\rotatebox{90}{\scriptsize{Wave frequency $[10^7 \text{rad/s}]$}}}
	\put(54,110){\scriptsize{Angle of incidence $[\text{rad}]$}}
	\put(243,290) {\scriptsize Transmission single, S incidence, free microstr.}
	\put(222,165) {\rotatebox{90}{\scriptsize{Wave frequency $[10^7 \text{rad/s}]$}}}
	\put(284,110){\scriptsize{Angle of incidence $[\text{rad}]$}}
	\end{picture}
\end{figure}

We conclude this subsection by showing the transmission coefficients at the ``single interface'', as a function of the frequency and the angle of incidence for L and S incident waves (Fig. \ref{fig:ReflSweepSingle}(a) and \ref{fig:ReflSweepSingle}(b), respectively). Comparing Fig. \ref{fig:ReflSlabSweepFreeL}(b) to \ref{fig:ReflSweepSingle}(a) and \ref{fig:ReflSlabSweepRMMS} to \ref{fig:ReflSweepSingle}(b), we can visualize the extent to which the single interface can be considered to represent a meta-structure of finite size. Although some basic averaged information is contained in Figures \ref{fig:ReflSweepSingle}(a) and \ref{fig:ReflSweepSingle}(b) (band-gap, critical angles), the detailed scattering behavior of the finite slab cannot be inferred from it. This provides additional evidence for the real need to propose a framework in which macroscopic boundary conditions can be introduced in a simplified way. Semi-infinite problems for metamaterials are solved in the context of homogenization methods in \cite{srivastava2017evanescent, willis2009exact}, but to the author's knowledge, a systematic method for the solution of scattering problems for metamaterials of finite size is not available in the literature. 

\section{Conclusions and further perspectives}\label{sec:conclusions}
In this paper we presented for the first time the scattering solution for a metamaterial slab of finite size, modeled via a rigorous micromorphic-type boundary value problem describing its homogenized behavior.

The correct macroscopic boundary conditions (continuity of macroscopic displacement, of generalized tractions and of micro-distortion/double-traction) are presented and are intrinsically compatible with the used macroscopic bulk PDEs. The scattering properties of the considered finite-size meta-structures, as obtained via the relaxed micromorphic model, are compared to a direct microstructured simulation. This latter simulation is obtained by assuming that the metamaterial's unit cell is periodic and linear-elastic. Excellent agreement is found for all angles of incidence and for frequencies going from the long-wave limit to the first band-gap and beyond. Further work will be devoted to applying the relaxed micromorphic model validated in this paper to the scattering of metamaterial obstacles that have finite size in both the $x_1$ and $x_2$ directions.

The results presented in this paper clearly indicate the need of a new mindset with respect to classical homogenization to open new possibilities for large-scale meta-structural design. These results will serve as a starting stepping stone for the development of new modeling tools for the mechanics of finite-sized metamaterials enabling the design of large-scale meta-structures.

\section*{Acknowledgements}
Angela Madeo and Domenico Tallarico acknowledge funding from the French Research Agency ANR, “METAS	MART”
	(ANR-17CE08-0006). Angela Madeo acknowledges support from IDEXLYON in the framework of the
	“Programme Investissement d’Avenir” ANR-16-IDEX-0005. All the authors acknowledge funding from the “Région
	Auvergne-Rhône-Alpes” for the “SCUSI” project for international mobility France/Germany.
\newpage

\appendix
\section{Appendix}\label{app:1}
\small
\subsection{Energy flux for the anisotropic relaxed micromorphic model}
\subsubsection{Derivation of expression \eqref{eq:fluxAniso}}\label{app:derivationflux}
The total energy is given by: 
\begin{equation}\label{eq:totalenergy}
E=J(u_{,t},\nabla u_{,t}, P_{,t})+W(\nabla u, P, \Curl P),
\end{equation}
where $J(u_{,t},\nabla u_{,t}, P_{,t})$ and $W(\nabla u, P, \Curl P)$ are defined in \eqref{eq:kineticRMM} and \eqref{eq:strainRMM}. Differentiating \eqref{eq:totalenergy} with respect to time\footnote{Let $\psi$ be a vector field and $A$ a second order tensor field. Then
	\begin{equation}\label{eq:vectorid1}
	\langle \nabla \psi, A\rangle  = \Div (\psi \cdot A) - \langle \psi, \Div A \rangle.
	\end{equation}
	Taking $\psi = u_{,t}$ and $A = \sigmatil$ we have
	\begin{equation}\label{eq:vectorid2}
	\left\langle \nabla u_{,t}, \sigmatil \right\rangle = \Div(u_{,t}\cdot \sigmatil) - \left\langle u_{,t},\Div \sigmatil \right\rangle.
	\end{equation}
	Furthermore, we have the following identity
	\begin{equation}\label{eq:vectorid3}
	\left\langle m, \Curl P_{,t} \right\rangle = \Div\left( (m^T\cdot P_{,t}):\epsilon\right)+\left\langle \Curl m , P_{,t}\right\rangle,
	\end{equation}
		which follows from the identity $\div( v \times w)=w \cdot\curl v - v \cdot \curl w$, where $v, w$ are suitable vector fields, $\times$ is the usual vector product and $:$ is the double contraction between tensors.} 
we have:
\begin{align}
E_{,t} =\, &\langle u_{,t},\rho u_{,tt} \rangle +\langle \sym P_{,t}, \Jmic\sym P_{,tt}\rangle + \langle  \skew P_{,t}, \Jc\skew P_{,tt}\rangle +\langle  \sym \nabla u_{,t},\Te\sym \nabla u_{,tt}\rangle \nonumber\\
& + \langle  \skew \nabla u_{,t},\Tc\skew \nabla u_{,tt}\rangle + \langle \Ce \sym(\nabla u - P),\sym (\nabla u - P)_{,t}\rangle + \langle \Cc \skew (\nabla u - P), \skew (\nabla u - P)_{,t}\rangle  \nonumber  \\
&+ \langle \Cmic \sym P, \sym P_{,t} \rangle + \mumic L_c^2\langle \Curl P, \Curl P_{,t}\rangle. \label{eq:EnergyTimeDerivative}
\end{align}
Using  the governing equations \eqref{eq:governingRMMclosed}, definitions \eqref{eq:relaxedRHS2} for $\widetilde{\sigma}, \widehat{\sigma}, s, m$ and \eqref{eq:vectorid1}, \eqref{eq:vectorid2}, \eqref{eq:vectorid3} we have:
\begin{itemize}
	\item[] 
	\begin{align*}
	\langle u_{.t},\rho \, u_{,tt}\rangle&=\langle u_{,t},\Div(\Te \sym \nabla u_{,tt}+\Tc \skew \nabla u_{,tt})+\Div\widetilde{\sigma}\rangle\\
		& = \langle u_{,t},\Div\widetilde{\sigma}\rangle+\langle u_{,t},\Div(\underbrace{\Te \sym \nabla u_{,tt}+\Tc \skew \nabla u_{,tt}}_{:=\sigmahat})\rangle \\
	&=\langle u_{,t},\Div\widetilde{\sigma}\rangle + \langle u_{,t},\Div\sigmahat\rangle,
	\end{align*}
	\vspace{-0.8cm}
	\item[] 
	\begin{align*}
	\langle \sym P_{,t}, \Jmic\sym P_{,tt}\rangle &+ \langle  \skew P_{,t}, \Jc\skew P_{,tt}\rangle\\
	&=\langle \sym P_{,t}+\skew P_{,t}, \Jmic\sym P_{,tt}\rangle + \langle  \sym P_{,t}+\skew P_{,t}, \Jc\skew P_{,tt}\rangle\\
	&=\langle P_{,t},\Jmic\sym P_{,tt}+ \Jc\skew P_{,tt}\rangle =\langle P_{,t},\widetilde{\sigma}_e -s - \sym \Curl m + \widetilde{\sigma}_c - \skew \Curl m \rangle\\
	&=\langle P_{,t}, \widetilde{\sigma} - s - \Curl m \rangle = \langle P_{,t},\sigmatil \rangle - \langle P_{,t}, s \rangle - \langle P_{,t}, \Curl m\rangle,
	\end{align*}
		\vspace{-0.8cm}
	\item[]
	\begin{align*}
	\langle  \sym \nabla u_{,t},\Te\sym \nabla u_{,tt}\rangle &+ \langle  \skew \nabla u_{,t},\Tc\skew \nabla u_{,tt}\rangle\\
	&=\langle  \sym \nabla u_{,t}+\skew \nabla u_{,t},\Te\sym \nabla u_{,tt}\rangle + \langle  \sym \nabla u_{,t}+\skew \nabla u_{,t},\Tc\skew \nabla u_{,tt}\rangle\\
	&=\langle \nabla u_{,t},\underbrace{\Te\sym \nabla u_{,tt}+\Tc\skew \nabla u_{,tt}}_{=\sigmahat}\rangle=\langle \nabla u_{,t},\sigmahat\rangle=
	\Div (u_{,t}\cdot \sigmahat) - \langle u_{,t},\Div \sigmahat \rangle,
	\end{align*}
		\vspace{-0.8cm}
		\item[]
	\begin{align*}
	\langle \Ce \sym(\nabla u - P),\sym (\nabla u - P)_{,t}\rangle  &+ \langle \Cc \skew (\nabla u - P), \skew (\nabla u - P)_{,t}\rangle\\
	&=\langle \Ce \sym(\nabla u - P),\sym (\nabla u - P)_{,t}+\skew (\nabla u - P)_{,t}\rangle\\
	&+ \langle \Cc \skew (\nabla u - P),\sym (\nabla u - P)_{,t} + \skew (\nabla u - P)_{,t}\rangle \\
	&=\langle \Ce \sym(\nabla u - P)+\Cc \skew (\nabla u - P),(\nabla u - P)_{,t}\rangle=\langle \sigmatil ,(\nabla u - P)_{,t}\rangle \\
	&=\langle \sigmatil, \nabla u_{,t}\rangle - \langle \sigmatil, P_t\rangle = \Div (u_{,t}\cdot \sigmatil) - \langle u_{,t},\Div \sigmatil \rangle - \langle \sigmatil, P_t\rangle,
	\end{align*}
	\vspace{-0.8cm}
	\item[]
	\begin{align*}
	\langle \Cmic \sym P, \sym P_{,t} \rangle = \langle \Cmic \sym P, \sym P_{,t} +  \skew P_{,t} \rangle = \langle \Cmic \sym P, P_{,t} \rangle =\langle s,P_{,t}\rangle,
	\end{align*}
	\vspace{-0.8cm}
	\item[]
	\begin{equation*}
		\mumic L_c^2\langle \Curl P, \Curl P_{,t}\rangle = \langle m, \Curl P_{,t} \rangle= \Div\left( (m^T\cdot P_{,t}):\epsilon\right) - \langle \Curl m, P_{t}\rangle
	\end{equation*}
	
\end{itemize}
So, by adding all the above and simplifying \eqref{eq:EnergyTimeDerivative} becomes:
\begin{equation}\label{eq:EnergyTimeDerivativeDiv}
E_{,t} = \Div \left[ \left(\sigmatil + \sigmahat \right)^T\cdot u_{,t} + \left( m^T \cdot P_{,t} \right):\epsilon \right],
\end{equation}
from which we can define the energy flux for the general anisotropic relaxed micromorphic model: 
\begin{equation}\label{eq:FluxAnisoapp}
H = -\left(\sigmatil + \sigmahat\right)^T\cdot u_{,t} - \left( m^T \cdot P_{,t} \right):\epsilon.
\end{equation}

\subsubsection{Analytical expression of the flux for the relaxed micromorphic model when $L_c=0$}\label{app:analyticalflux}
The flux $H$ of the relaxed micromorphic model when $L_c=0$ can be written as (using Lemma \ref{Lemma1})
\begin{equation}\label{eq:fluxrmm1}
H = \frac{1}{2}\Re \left[
\left( \alpha_1 \omega (-2 \omega^2 \mathcal{A} + \mathcal{B}) + \alpha_2 \omega(- 2 \omega^2 \mathcal{C} + \mathcal{D})	\right) \mathcal{E}^{\star}  +\left(
\alpha_1\omega \,( \mathcal{F} + \omega \, \mathcal{G} + \omega^2 \mathcal{H}) +\alpha_2 \omega (\mathcal{I} + \omega \, \mathcal{J} + \omega^2 \mathcal{K})
\right) \mathcal{L}^{\star}
\right],
\end{equation}
\normalsize
with
\begin{align}
\mathcal{A} &= k_1^{\op} \phi_1^{\op} \left(\bar{\eta}_1+\frac{1}{2}\bar{\eta}_3\right) + k_0 \,\phi_2^{\op} \bar{\eta}_3, \nonumber\\
\mathcal{B} &=(\lame + 2\mue)\left(k_1^{\op} \phi_1^{\op} - \phi_3^{\op}\omega\right) + \lame \left(k_0\, \phi_2^{\op}-\phi_6^{\op}\omega\right),\nonumber\\
\mathcal{C} &=  k_1^{\twop} \phi_1^{\twop} \left(\bar{\eta}_1+\frac{1}{2}\bar{\eta}_3\right) + k_0\, \phi_2^{\twop} \bar{\eta}_3,\nonumber\\
\mathcal{D} &=(\lame + 2\mue)\left(k_1^{\twop} \phi_1^{\twop} - \phi_3^{\twop}\omega\right) + \lame \left(k_0\, \phi_2^{\twop}-\phi_6^{\twop}\omega\right),\nonumber\\
\mathcal{E} &= \alpha_1 \phi_1^{\op} + \alpha_2 \phi_1^{\twop},\nonumber\\
\mathcal{F} &= k_0 \, \left( \mue^{\ast} -\muc \right) \phi_1^{\op} + k_1^{\op}  \left( \mue^{\ast} +\muc \right) \phi_2^{\op}, \\
\mathcal{G} &=-\left(\mue^{\ast} -\muc\right) \phi_4^{\op} - \left(\mue^{\ast} +\muc\right) \phi_5^{\op},\nonumber\\
\mathcal{H} &= k_0 \,\phi_1^{\op}\left(\frac{1}{4} \bar{\eta}_2-\bar{\eta}^{\ast}\right) -  k_1^{\op} \phi_2^{\op}\left(\frac{1}{4} \bar{\eta}_2 + \bar{\eta}^{\ast}\right),\nonumber\\
\mathcal{I} &= k_0 \, \left( \mue^{\ast} -\muc \right) \phi_1^{\twop} + k_1^{\twop}  \left( \mue^{\ast} +\muc \right) \phi_2^{\twop}, \nonumber\\
\mathcal{J} &= -\left(\mue^{\ast} -\muc\right) \phi_4^{\twop} - \left(\mue^{\ast} +\muc\right) \phi_5^{\twop}, \nonumber\\
\mathcal{K} & =k_0 \,\phi_1^{\twop}\left(\frac{1}{4} \bar{\eta}_2-\bar{\eta}^{\ast}\right) -  k_1^{\twop} \phi_2^{\op}\left(\frac{1}{4} \bar{\eta}_2 + \bar{\eta}^{\ast}\right),	\nonumber\\
\mathcal{L} &=  \alpha_1 \phi_2^{\op} + \alpha_2 \phi_2^{\twop}.\nonumber
\end{align}

\subsection{The matrix $\widehat{A}$}\label{app:matrixA}
We present the matrix row-wise. We have 
\small
\begin{align*}
\widehat{A}_{1} =  \left(  
\begin{array}{c}
-\left(\rho + k_1^2 (2\bar{\eta}_1 + \bar{\eta}_3) + k_2^2 \left( \frac{1}{4}\bar{\eta}_2+\bar{\eta}_1^{\ast} \right)\right)\,\omega^2 +k_1^2(\lame +2 \mue ) +k_2^2 \left( \muc +\mue^{\ast} \right)\\
k_1 k_2 \left(\frac{1}{4} \bar{\eta}_2 - \bar{\eta}_3+\bar{\eta}_1^{\ast}\right)\,\omega^2 + k_1 k_2  (\lame - \muc +\mue^{\ast})\\
i\, k_1 (\lame + 2\mue)\\
 i\, k_2(\muc + \mue^{\ast})\\
 i\, k_2(-\muc + \mue^{\ast})\\
 i\,k_1 \lame
\end{array}
\right)^T,\\
 \widehat{A}_{2} = \left(  
\begin{array}{c}
k_1k_2  \left(\frac{1}{4}\bar{\eta}_2 - \bar{\eta}_3 -\bar{\eta}_1^{\ast}\right)\,\omega^2 + k_1k_2\left(\lame -\muc + \mue^{\ast}\right)\\
-\left(\rho + k_2^2 (2\bar{\eta}_1+\bar{\eta}_3) +k_1^2 \left(\frac{1}{4}\bar{\eta}_2+\bar{\eta}_1^{\ast}\right) \right)\,\omega^2 + k_2^2(\lame + 2\mue )+k_1^2\left( \muc+\mue^{\ast} \right)\\
i\,k_2 \lame \\
-i \, k_1 (\muc-\mue^{\ast})\\
i\,k_1 (\mue + \mue^{\ast})\\
i\,k_2(\lame + 2\mue)
\end{array}
\right)^T, 
\end{align*}
\begin{align*}
\begin{array}{cc}
\widehat{A}_{3} = \left(  
\begin{array}{c}
-i\,k_1(\lame + 2 \mue)\\
-i\, k_2 \lame\\
-(2\eta_1 + \eta_3)\,\omega^2 +\lame + \lammic + 2 (\mue + \mumic) + k_2^2\Lc\\
-k_1k_2\Lc \\
0\\
- \eta_3 \,\omega^2 + \lame + \lammic 
\end{array}
\right)^T,&  \widehat{A}_{4} =  \left(  
\begin{array}{c}
-i\,k_2(\muc + \mue^{\ast})\\
i\,k_1 (\muc - \mue^{\ast})\\
-k_1k_2\Lc\\
- (\eta_2 + \eta_1^{\ast})\,\omega^2 + \muc + \mue^{\ast} + \mumic^{\ast} + k_1^2\Lc\\
 (\eta_2 - \eta_1^{\ast})\,\omega^2 - \muc + \mue^{\ast} + \mumic^{\ast} \\
0
\end{array}
\right)^T, 
\end{array}
\end{align*}
\begin{align*}
\begin{array}{cc}
\widehat{A}_{5} = \left(  
\begin{array}{c}
i\,k_2(\muc-\mue^{\ast})\\
-i\,k_1(\muc + \mue^{\ast})\\
0\\
(\eta_2 - \eta_1^{\ast})\,\omega^2 -\muc + \mue^{\ast} + \mumic^{\ast} \\
-(\eta_2 + \eta_1^{\ast})\,\omega^2 + \muc + \mue^{\ast} + \mumic^{\ast} + k_2^2\Lc\\
-k_1k_2\Lc
\end{array}
\right)^T,&  \widehat{A}_{6} =  \left(  
\begin{array}{c}
-i\,k_1\lame\\
-i\,k_2(\lame+2\mue)\\
-\eta_3\,\omega^2 + \lame + \lammic \\
0\\
-k_1k_2\Lc\\
-(2\eta_1+\eta_3)\,\omega^2 + \lame + \lammic + 2(\mue + \mumic) + k_1^2\Lc
\end{array}
\right)^T, 
\end{array}
\end{align*}
Then, the matrix $\widehat{A}$ is 
\begin{equation}
\widehat{A} = \left( \widehat{A}_1,\widehat{A}_2,\widehat{A}_3,\widehat{A}_4,\widehat{A}_5,\widehat{A}_6\right)^T
\end{equation}

\section{Lemma \ref{Lemma1}}\label{app:lemma}
\small
We have the following well-known result.
\begin{lemma}\label{Lemma1}
	Let 
	\[
	u(x,t) = A(x)e^{i(\langle k,x\rangle -\omega t)}, \quad \sigma(x,t) = B(x)e^{i(\langle k,x\rangle -\omega t)}
	\]
	be two functions with $A,B: \R^3 \to \C$. Then the following holds
	\begin{equation}\label{eq:Lemma1}
	\frac{1}{T} \int_0^{T} \Re\{u(x,t)\}  \Re\{\sigma(x,t)\} dt= \frac{1}{2} \Re (u(x,t) \sigma^{*}(x,t)),
	\end{equation}
	where $T=\frac{2\pi}{\omega}$ is the period of the functions $u, \sigma$, ${*}$ denotes the complex conjugate, $x=(x_1,x_2)^{\rm T}$ and $k=(k_1,k_2)^{\rm T}$.
\end{lemma}
\begin{proof}
	We have:
	\begin{align}
	\frac{1}{T}\int_0^{T} \Re \{u(x,t)\} \Re \{\sigma(x,t)\}dt =& \frac{1}{T}\int_0^{T} \Re \left(A e^{i(\langle k,x\rangle -\omega t)}\right) \Re \left(B e^{i(\langle k,x\rangle -\omega t)}\right)dt \nonumber\\
	=&\frac{1}{T}\int_0^{T} \frac{A e^{i(\langle k,x\rangle -\omega t)} + A^{*} e^{-i(\langle k,x\rangle -\omega t)}}{2}\frac{B e^{i(\langle k,x\rangle -\omega t)}+B^{*} e^{-i(\langle k,x\rangle -\omega t)}}{2}dt \nonumber \\
	=&\frac{1}{4T}\int_0^{T} A\,B e^{i2\langle k,x\rangle } \underbrace{e^{-i2\omega t}}_{\text{periodic}} + A B^{*}+ A^{*} B  + A^{*} B^{*} e^{-i2\langle k,x\rangle} \underbrace{e^{i2\omega t}}_{\text{periodic}} dt \nonumber \\
	=&\frac{1}{4T}\int_0^{T} AB^{*}+ A^{*}B   dt \nonumber \\
	=& \frac{1}{2} \Re \left(AB^{*}  \right) \nonumber \\
	=& \frac{1}{2} \Re \left(u(x,t) \sigma^{*}(x,t)\right),\label{eq:ALemmaProof}
	\end{align}
	since 
	\[\frac{1}{2} \Re \left(u(x,t) \sigma^{*}(x,t)\right) = \frac{1}{2} \Re \left( Ae^{i(\langle k,x\rangle -\omega t)}  B^{*}e^{-i(\langle k,x\rangle -\omega t)}   \right) =\frac{1}{2} \Re \left(AB^{*} \right) \]
	where we used the facts that the integral of the periodic function $e^{2i\omega t}$ over its period is zero and that for any complex number $z\in \C$: $\Re(z) = \frac{z+z^{*}}{2}$.
\end{proof}

{\footnotesize{}\bibliographystyle{plainnat}
		\bibliography{library.bib}
}
\end{document}